\pdfoutput=1
\documentclass{amsart}
\usepackage{breakcites}
\usepackage{preamble}
\usepackage{macros}
\usepackage{categories}
\usepackage{jon-tikz}

\def\paragraph{\subsubsection*}

\title{Normalization for Multimodal Type Theory}
\author{Daniel Gratzer}
\date{\today}

\begin{document}
\begin{abstract}
  We consider the conversion problem for multimodal type theory
  (\MTT{})~\parencite{gratzer:journal:2020} by characterizing the normal forms of the type theory
  and proving normalization (\cref{thm:normalization:normalization}). Normalization follows from a
  novel adaptation of Sterling's Synthetic Tait
  Computability~\parencite{sterling:modules:2020,sterling:phd} which generalizes
  the framework to accommodate a type theory with modalities and multiple modes. As a corollary of
  our main result, we reduce the conversion problem of \MTT{} to the conversion problem of its mode
  theory (\cref{cor:normalization:conversion}) and show the injectivity of type constructors
  (\cref{cor:normalization:pi-inj}).
\end{abstract}

\maketitle
\tableofcontents

\section{Introduction}
\label{sec:intro}

The last twenty years of development of type theory has seen many different extensions of type
theories to account for different modalities. Most of these type theories are
specialized---accounting for a specific collection of modalities---and intended for a particular
model. This specialization has allowed for concise and practicable syntax in some cases, but has
also created a tremendous amount of churn where even the smallest change to the modal situation
requires substantial work to produce a new type theory.

In response to this situation, several ``frameworks'' for modal type theory have been put
forward. A framework for modal type theory should enjoy the following:
\begin{enumerate}
\item sufficient generality to accommodate a large class of modalities,
\item practicable syntax in specific applications,
\item a wide variety of theorems which hold for instantiations.
\end{enumerate}

There is a degree of tension between these three goals: a framework which is more general often
admits more cumbersome syntax, and wider classes of type theories typically enjoy fewer shared
properties. Recently, \textcite{gratzer:2020} introduced \MTT{}, a type theory parameterized by a
collection of modalities. \MTT{} aims to support a highly usable syntax for any collection of
modalities which behave like (dependent) right adjoints~\parencite{birkedal:2020}. This proves to be
a rich class of modalities, allowing encodings of type theories for guarded recursion, internal
parametricity, axiomatic cohesion, \etc{}

\subsection{\MTT{}}

We briefly recall some of details of \MTT{}. For a full account, see the extended submission on
\MTT{}~\parencite{gratzer:journal:2020}. As mentioned, \MTT{} is parameterized by a collection of
modalities. Formally, \MTT{} is defined over a \emph{mode theory}~\parencite{licata:2016}: a strict
2-category whose objects represent modes, while 1-cells determine modalities and 2-cells natural
transformations between them.

Given a mode theory $\Mode$, each mode $m : \Mode$ determines a mode in \MTT{}. Each mode behaves
somewhat like its own type theory. There are separate judgments for each mode $\IsCx{\Gamma}$,
$\IsTy{A}$, and $\IsTm{M}{A}$ and types in mode $m$ are closed under the standard connectives in
Martin-L{\"o}f type theory (dependent products, sums, booleans, an intensional identity type, and a
universe).

Different modes only interact through modalities. Given a morphism $\Mor[\mu]{n}{m}$, there is an
induced modality $\Modify{-}$ in \MTT{} which sends types from mode $n$ to mode $m$. This modality
behaves like a right adjoint~\parencite{birkedal:2020}, whose left adjoint is given an action only
on contexts $\LockCx{-}$.\footnote{Other presentations of \MTT{} and Fitch-style modal type theories
  write $-.\Lock_{\mu}$ for this action, but we find $\LockCx{-}$ more uniform and less cumbersome.}
This operation is a proper functor between context categories, with an induced action on
substitutions $\LockSb{\gamma}$. As the left adjoint, $\LockCx{-}$ sends context from mode $m$ to
mode $n$. This left adjoint is used to give the formation and introduction rules for \MTT{} by
transposition:
\begin{mathpar}
  \inferrule{
    \IsTy[\LockCx{\Gamma}]{A}<n>
  }{
    \IsTy{\Modify{A}}
  }
  \and
  \inferrule{
    \IsTm[\LockCx{\Gamma}]{M}{A}<n>
  }{
    \IsTm{\MkBox{M}}{\Modify{A}}
  }
\end{mathpar}

Two questions remain: how does one eliminate an element of $\Modify{A}$ and how can one use a
context of the form $\LockCx{\Gamma}$. While related modal type theories solve both issues through
defining the elimination rule to be transposition in the opposite direction, this causes issues for
the substitution property. Instead, \MTT{} employs an alternative approach. First, variables in the
context is annotated with a modality:
\begin{mathpar}
  \inferrule{
    \IsCx{\Gamma}
    \\
    \IsTy[\LockCx{\Gamma}]{A}<n>
  }{
    \IsCx{\ECx{\Gamma}{A}}
  }
\end{mathpar}

These modal annotations are used to determine when a variable can be accessed; a variable annotated
with $\mu$ can be accessed precisely when it is behind $\LockCx{-}$:
\begin{mathpar}
  \inferrule{ }{
    \IsTm[\LockCx{\ECx{\Gamma}{A}}]{\Var{0}}{\Sb{A}{\LockSb{\Wk}}}
  }
\end{mathpar}

Next, the elimination rule for $\Modify{-}$ is used to bridge the gap between a variable annotated
with $\nu \circ \mu$ of type $A$ and a variable annotated with $\nu$ of type $\Modify{A}$. There is
a canonical substitution $\Mor{\ECx{\Gamma}{A}<\nu\circ\mu>}{\ECx{\Gamma}{\Modify{A}}<\nu>}$, but a
priori no map in the reverse direction. While we do not provide an inverse outright, the elimination
rule ensures that these two contexts are isomorphic ``from the perspective of a
type''. Semantically, one ensures that display maps are weakly right orthogonal to this map, but
more concretely we add the following elimination principle:
\begin{mathpar}
  \inferrule{
    \Mor[\nu]{o}{n}
    \\
    \Mor[\mu]{n}{m}
    \\\\
    \IsCx{\Gamma}
    \\
    \IsTy[\LockCx{\LockCx{\Gamma}}<\nu>]{A}<o>
    \\
    \IsTy[\ECx{\Gamma}{\Modify[\nu]{A}}]{B}
    \\
    \IsTm[\LockCx{\Gamma}]{M_0}{\Modify[\nu]{A}}<n>
    \\
    \IsTm[\ECx{\Gamma}{A}<\mu \circ \nu>]{M_1}{\Sb{B}{\ESb{\Wk}{\MkBox[\nu]{\Var{0}}}}}
  }{
    \IsTm{\LetMod{M_0}{M_1}<\nu>[\mu]}{\Sb{B}{\ESb{\ISb}{M_0}}}
  }
\end{mathpar}

Finally, the interactions between modalities are codified by organizing the collection of functors
$\LockCx{-}<->$ into a single 2-functor out of $\Coop{\Mode}$. Concretely, this ensures that there
is a substitution $\IsSb[\LockCx{\Gamma}<\mu>]{\Key{\alpha}{\Gamma}}{\LockCx{\Gamma}<\nu>}$ for each
2-cell $\Mor[\alpha]{\nu}{\mu}$ and that a variety of functoriality equations hold \eg{}:
\[
  \EqCx{\LockCx{\Gamma}<\mu\circ\nu>}{\LockCx{\LockCx{\Gamma}}<\nu>}<o>
  \qquad
  \EqCx{\LockCx{\Gamma}<\ArrId{m}>}{\Gamma}
\]

2-functoriality is sufficient to ensure that the modal types organize into a weak 2-functor, with
type-theoretic equivalences between \eg{}, $\Modify[\ArrId{m}]{A}$ and $A$.

\subsection{Normalization for \MTT{}}

Returning to the criterion for a framework for modal dependent type theories,
\textcite{gratzer:journal:2020} showed that \MTT{} enjoys substantial metatheoretic properties. In
particular, every instantiation of \MTT{} is sound and enjoys a canonicity result. Both of these
results are proven semantically, taking advantage of the fact that an instantiation of \MTT{} comes
equipped with a well-behaved class of models in which syntax is initial. Notably, however, they
stopped short of showing that \MTT{} enjoys a normalization result, and the lack of such a result
precludes a general approach to implementing instantiations of \MTT{}.

Here, we rectify the situation and show that \MTT{} admits a normalization result. As a corollary,
conversion in \MTT{} is decidable when equality of modalities and 2-cells is decidable. This result
paves the way for an implementation of \MTT{} which does not necessitate a full reimplementation for
each tweak of the modalities.

Normalization results for dependent type theories are typically involved affairs, requiring the
construction of a sophisticated PER model, families of logical relations, and other technical
constructions. For this proof, however, we adapt recent gluing arguments for normalization to give a
concise and detailed
proof~\parencite{altenkirch:1995,fiore:2002,altenkirch:2016,coquand:2019,sterling:phd}.

While the literature contains a large number of modal dependent type theories, very few are proven
to enjoy normalization. Of these, the most closely related is
\MLTTLock{}~\parencite{gratzer:2019}. This is a type theory extending \MLTT{} with an idempotent
comonad. However, compared to the normalization proof of \MLTTLock{} given by
\textcite{gratzer:2019}, this result is significantly shorter and more conceptual.\footnote{As a
  crude measurement, the proof of normalization for \MLTTLock{} occupies the bulk of the
  accompanying 90 page technical report~\parencite{gratzer:tech-report:2019}. The proof of
  normalization for \MTT{}, by contrast, takes approximately 25 pages to present with more detail.}
Furthermore, the result for \MTT{} applies to a wide variety of modal situations, including modal
type theory with an idempotent comonad.

\subsection{Proof outline}

In broad strokes, the proof of normalization for \MTT{} proceeds in several stages. First, we
introduce the technical device of \MTT{} cosmoi in \cref{sec:models}. An \MTT{} cosmos is a more
flexible version of the models of \MTT{} originally described. In particular, cosmoi are only
required to be locally Cartesian closed, so that there is no `fiberwise representability'
requirement, and morphisms between models are only required to preserve several structures up to
isomorphism. This definition of model is more closely related to the presentations of type theories
as \emph{representable map categories}~\parencite{uemura:2019} or \emph{locally Cartesian closed
  categories}~\parencite{gratzer:lcccs:2020}. This flexibility is essential for the application of
synthetic Tait computability and allows us to avoid tedious calculations.

\begin{remark}
  We emphasize, however, this change of categories is a technical detail of the proof: the
  normalization function is defined on the stricter syntax constituting an initial object of the CwF
  models as originally defined rather than the initial object of the category of cosmoi. While it is
  more work to obtain a normalization result of this form, it ensures that one does not need to
  accept the initial \MTT{} cosmos as syntax; such acceptance would require a leap of faith, because
  there is no adequacy result for \MTT{} cosmoi corresponding those proven by \textcite{uemura:2019}
  and \textcite{gratzer:lcccs:2020}. Such an adequacy result is difficult to formulate for \MTT{}
  cosmoi because of the lack of established modal logical frameworks.
\end{remark}

In \cref{sec:foundations} we introduce the necessary mathematics to generalize synthetic Tait
computability (STC) from working internally to a single gluing category to applying to a collection
of gluing categories interconnected by functors and natural transformations. In particular, we show
that one may work internally to this network of categories in \MTT{}, using \MTT{} modalities to
pass between gluing categories. We also show that the fibered modalities available in each gluing
category commute with modalities already present in \MTT{}, a crucial result for constructing the
normalization model.

With this framework in place, \cref{sec:renamings} and \cref{sec:prereq} define the basic objects of
the normalization model: renamings, normal and neutral forms, \etc{} In particular,
\cref{sec:prereq} defines the gluing categories used for normalization, along with the functors
between them. This portion of the proof proceeds with minimal alternations from a standard proof of
normalization for \MLTT{}. The only changes are to account for multiple modalities, and this change
is mostly one of book-keeping. One novelty of \cref{sec:renamings} is the usage of a non-free
presentation of renamings. By allowing renamings to enjoy non-trivial equations, we vastly simplify
the development. As a consequence of this decision, we only must show that the action of a renaming
on a normal form respects these equations --- a straightforward calculation.

The actual construction of the normalization model, and the heart of the proof, takes place in
\cref{sec:normalization-model}. In particular, we define an \MTT{} cosmos of glued categories
together with the reify and reflect maps needed to define the normalization function. Those familiar
with proofs of normalization by gluing, and in particular with those given by STC will find much of
the construction familiar. The central challenge is the construction of modal types in the
normalization model (\cref{lem:normalization-model:mod}), and in particular the construction of the
modal elimination law. Other connectives which do not interact with modalities (booleans, universes,
dependent sums, \etc{}) can be constructed in this multimodal setting with no alteration.

Finally, the normalization function is extracted from this model in \cref{sec:normalization}. Taking
advantage of the fact that this proof is constructive, we also observe that the normalization
function is effective. This uses a novel approach to derive a normalization result for strict
syntax, rather than the initial \MTT{} cosmos. Using these results, we establish several other
corollaries such as the decidability of equality and the injectivity of type-constructors.


\section{Normalization via Synthetic Tait Computability}
\label{sec:stc}

The proof of normalization for \MTT{} follows from a systematic generalization of Synthetic Tait
Computability (STC) to multimodal type theories. While a full introduction of the theory of STC is
beyond the scope of this work, given that a comprehensive account has not yet emerged\footnote{We
  hope that the forthcoming work of \textcite{sterling:phd} will serve this role} we give a short
summary of the key ideas here. We focus in particular on proving normalization via STC. Other
descriptions of STC are given by \textcite{sterling:modules:2020}, \textcite{sterling:stc:2020}, and
\textcite{sterling:2021}.

Prior to discussing STC, we remark that in the process of embracing STC we are naturally led to a
number of ideas which are not intrinsically categorical but still crucial to the concision of the
proof.

\begin{itemize}
\item Our normalization proof is \emph{reduction-free}, and
\item the algorithm works only over equivalences classes of well-typed terms.
\end{itemize}

More explicitly, our proof is reduction-free in that it does not proceed by fixing a rewriting
system presenting the equational theory which we then prove to be confluent and strongly
normalizing. This avoids the thorny issue of finding an rewriting system which faithfully presents
the equational theory of \MTT{}---the inclusion of $\eta$-laws for dependent sums greatly
complicates such a task.

To the second point, working with terms only up to definitional equality obviates the need for a
separate proof of \emph{completeness} of the normalization algorithm. More than this, it is also a
necessary step to treating \MTT{} categorically; without this quotienting, the category of contexts
is not even a category. Similarly, the universal properties of various connectives of \MTT{} are
crucially leveraged to simplify aspects of the normalization proof, but these universal properties
only come into being after taking terms up to definitional equality.

\subsection{Normalization by gluing}
\label{sec:stc:nbg}

Proving normalization via STC involves working internally to a particular glued category. This idea
branches off the observation that that logical relations arguments could be systematically recast as
the construction of a model of type theory in a \emph{Freyd cover} or \emph{Sierpinski cone
  (scone)}~\parencite{mitchell:1993}. This is a special case of the more general gluing
construction:

\begin{definition}
  Given a functor $\Mor[F]{\CC}{\DD}$, the gluing category $\GL{F}$ is the comma category
  $\COMMA{\ArrId{\DD}}{F}$. Explicitly, objects are triples $(D, C, \Mor[f]{D}{F(C)})$ and commuting
  squares between them.
\end{definition}

\begin{definition}
  Given a category $\CC$, The Sierpinski cone $\SIERP{\CC}$ is $\GL{\Hom[\CC]{\ObjTerm{\CC}}{-}}$.
\end{definition}

To a rough approximation, a logical relation can be seen as an assignment of types to predicates on
their closed terms. Taking $\TT$ to be the category of contexts and substitutions between them and
blurring the distinction between contexts and types as is often done in the simply-typed case, this
is more or less the content of an object of $\SIERP{\TT}$; a triple
$(S, A, \Mor[f]{S}{\Hom{\ObjTerm{}}{A}})$ can be viewed as a predicate on the closed elements of $A$
with $P(a) = f^{-1}(a)$. Notice that the correspondence is imperfect because there may be multiple
distinct elements in the fiber $P(a)$. In fact, $f^{-1}(-)$ determines a \emph{proof-relevant}
predicate on closed elements of $A$. This can be fixed by cutting $\SIERP{\TT}$ down to consist of
objects $(S, A, \EmbMor[f]{S}{\Hom{\ObjTerm{}}{A}})$, but this is both less categorically natural
and, as we shall see in later, counterproductive.

What makes this shift in perspective so useful is the remarkable ability of $\SIERP{\CC}$ to
seemingly inherit all the structure of $\CC$ in such a way that the projection
$\Mor{\SIERP{\CC}}{\CC}$ preserves this structure. For instance, if $\CC$ is Cartesian closed, then
so too is $\SIERP{\CC}$ and the projection preserves all this structure. In fact, an explicit
construction of \eg{}, the exponential in $\SIERP{\CC}$ yields almost precisely the standard
construction of the logical relation at function type. The benefit of relying on such technology,
however, is that one can typically avoid these explicit calculations; once it is known that
$\SIERP{\CC}$ is Cartesian closed, canonicity at ground type can be proven without ever needing to
explicitly calculate the construction of exponentials.

The correspondence between logical relations and gluing was pushed further in multiple
directions~\parencite{altenkirch:1995,streicher:1998,altenkirch:2001,fiore:2002}. In particular, the
insistence of logical relations of \emph{closed terms} of a type can be relaxed by switching to a
different comma category. We are interested in proving normalization, emphatically not a result
solely concerned with closed terms.

Accordingly, rather than considering a model in $\SIERP{\CC}$, we will consider a glued category
valued in presheaves over renamings. It is perhaps more natural to expect one to examine presheaves
over contexts: after all, terms organize into a presheaf over arbitrary contexts and substitutions,
with substitution giving rise to the reindexing action of the presheaf. One would then consider a
model in $\GL{\Yo}$; the category given by gluing along the Yoneda embedding.

The first attempts at categorical proofs of normalization constructed models in $\GL{\Yo}$, so that
one considered predicates not on the closed terms of type $A$, but over the family
$(\Hom{\Gamma}{A})_{\Gamma}$. The presheaf condition corresponds to the monotonicity conditions
familiar to all Kripke logical relations, but in particular those for normalization. Unfortunately,
these presheaf conditions are too onerous to view normal forms as a presheaf on contexts --- normal
forms are never stable under all substitutions --- so that the resultant normalization function is
given as a map from terms to terms and indistinguishable from the identity function on terms.

This can be rectified taking presheaves over the wide subcategory of renamings $\RR \subseteq
\TT$. By cutting down the substitutions allowed between contexts, normal forms may be organized into
an object of $\PSH{\RR}$~\parencite{altenkirch:1995,fiore:2002}. Accordingly, normalization proceeds
not by gluing along the Yoneda lemma as still repeated in literature, but by gluing along a
restriction of the Yoneda lemma, the nerve $\Mor[\Nv]{\TT}{\PSH{\RR}}$. Unlike canonicity, it is
insufficient to simply construct a model of $\TT$ in $\GL{\Nv}$, extra structure is required which
`sandwiches' the predicate for a type $A$ between the neutral and normal forms of $A$, so that the
following commutative diagram exists in $\PSH{\RR}$:
\begin{equation}
  \begin{tikzpicture}[diagram]
    \node (Ne) {$\Ne{A}$};
    \node [right = of Ne] (Pred) {$P_A$};
    \node [right = of Pred] (Nf) {$\Nf{A}$};
    \node [below = 1.5cm of Pred] (A) {$\Nv{A}$};
    \path[->] (Ne) edge (A);
    \path[->] (Ne) edge node[above] {$\Reflect{A}$} (Pred);
    \path[->] (Pred) edge (A);
    \path[->] (Pred) edge node[above] {$\Reify{A}$} (Nf);
    \path[->] (Nf) edge (A);
  \end{tikzpicture}
\end{equation}

The top two maps of this diagram are categorical realizations of the \emph{reflect} and \emph{reify}
maps found in normalization-by-evaluation~\cite{abel:2013} and are used to extract the normalization
function. In the simply-typed case, these maps can be constructed by induction after the fact, but
in the dependently-typed setting this approach does not scale and the two must be constructed
simultaneously. Accordingly, we require a type in our glued model to be not just a proof-relevant
predicate $P_A$ over $\Nv{A}$, but a triple of
$(P_A, \Reflect{A}, \Reify{A})$~\cite{altenkirch:1995,fiore:2002}.

One final remark is needed prior to the introduction of STC. Thus far we have limited consideration
to simply-typed languages. The full utility of proof-relevance only becomes apparent when
considering dependent type theory with universes. In these cases, it becomes necessary to associate
with an element of a universe not merely a proposition (computable or not computable) but to instead
pair an element with the \emph{computability data of a type} again. This is beautifully handled by
the proof-relevance of predicates: the predicate on universes at $A : \Uni$ contains choices of
computability data over $\Dec{A}$ as witnesses~\parencite{shulman:2015,coquand:2019}.

\subsection{Synthetic Tait Computability}
\label{sec:stc:stc}

Finally, we now turn to Synthetic Tait Computability. STC builds on the idea of working with glued
categories, but instead of gluing together locally Cartesian closed categories presenting dependent
type theories~\parencite{gratzer:lcccs:2020}, STC is concerned with topoi. Suppose $\TT$ presents a
type theory, and $\RR \subseteq \TT$ is the subcategory of renamings. Rather than gluing together
the locally Cartesian closed categories $\TT$ and $\PSH{\RR}$ along $\Nv$, we instead glue together
$\PSH{\TT}$ and $\PSH{\RR}$ along the inverse image of the essentially geometric morphism induced by
the inclusion $\Mor[i]{\RR}{\TT}$. As topoi are closed under gluing, the resultant glued category
$\GL{i}$ is a topos. In fact, by a standard result~\parencite{sga:4,carboni:1995} it is a presheaf
topos.

As a result, $\GL{i}$ supports a model of extensional Martin-L\"of type theory with a hierarchy of
cumulative universes. Inside this type theory, two fibered modalities, $\Open$ and $\Closed$, allow
us to recover $\PSH{\TT}$ and $\PSH{\RR}$ inside of $\GL{i}$. By relaxing the definition of a model
of Martin-L{\"o}f type theory to non-representability universes moreover, we may specify and
construct the glued model purely internally to $\GL{i}$.

This internal approach offers a significant improvement over previous proofs of gluing for dependent
type theory~\parencite{coquand:2019,kaposi:gluing:2019} which work externally and incur what
\textcite{sterling:phd} has referred to as an `avalanche' of naturality obligations. As we shall see
in \cref{sec:normalization-model}, these tedious calculations are replaced with programming
exercises in extensional type theory (or \MTT{} in our case). In particular, a model lying over the
model in $\PSH{\TT}$ corresponds to a series of constants subject to the requirement that they are
sent by $\Open$ to the constants defining the model in $\PSH{\TT}$.

We refer to the internal type theory of $\GL{i}$---extensional type theory with cumulative
universes, two fibered modalities, and a few axioms governing their behavior---as the language of
Synthetic Tait Computability.


\section{\MTT{} Cosmoi}
\label{sec:models}

While the basic theory of models of \MTT{} is detailed by \textcite{gratzer:journal:2020}, we
generalize to locally Cartesian closed categories. This change, inspired by
\textcite{gratzer:lcccs:2020}, is crucial to actually carrying out the normalization theorem as it
removes \emph{representability requirements} from the universes of types and terms. These prove
superfluous for the theorem, and forcing the glued model to satisfy them is both unnatural and
surprisingly challenging.

\begin{remark}
  Henceforth, we shall consider \MTT{} over a fixed mode theory $\Mode$.
\end{remark}

\begin{definition}[Definition 18~\parencite{awodey:2018}]
  An internal lifting structure $s : i \pitchfork \tau$ between a pair of morphisms $\Mor[i]{A}{B}$
  and $\Mor[\tau]{X}{Y}$ is a section of canonical map $\Mor{X^B}{Y^B \times_{Y^A} X^A}$.
\end{definition}

\begin{definition}
  The 2-category $\CAT_g$ consists of small categories, functors between them, and invertible
  natural transformations.
\end{definition}

\begin{definition}
  \label{def:cosmoi:cosmos}
  The locally full 2-subcategory $\VV$ of $\Hom{\Mode}{\CAT_g}$ is spanned by pseudofunctors $F$
  enjoying the following two properties:
  \begin{itemize}
  \item For each $m : \Mode$, $F(m)$ is a locally Cartesian closed category.
  \item For each $\Mor[\mu]{n}{m}$, $F(\mu)$ is a right adjoint.
  \end{itemize}
  Note that we do not require that $F(\mu)$ preserves the locally Cartesian closed structure of
  $F(m)$. We write $\LKan{F(\mu)}$ for the left adjoint of $F(\mu)$.
\end{definition}

\begin{definition}
  \label{def:cosmoi:structured-cosmos}
  An object $F : \VV$ is an \MTT{} cosmos when equipped with the following structure:
  \begin{enumerate}
  \item In $F(m)$, there is a universe $\Mor[\El{m}]{\EL{m}}{\TY{m}}$ with a choice of codes
    witnessing its closure under dependent sums, and booleans. Additionally, there is a choice of
    code making $\El{m}$ closed under modal dependent products: a code for
    $\Poly{f_\mu(\El{n})}(\El{m})$ for each $\Mor[\mu]{n}{m}$.
  \item For each $\mu$, there exists a chosen commuting square
    \begin{equation}
      \DiagramSquare{
        nw = F(\mu)(\EL{n}),
        sw = F(\mu)(\TY{n}),
        ne = \EL{m},
        se = \TY{m},
        south = \mathbf{Mod},
        width = 4cm,
        height = 2cm,
      }
      \label[diagram]{diag:models:modal-intro}
    \end{equation}
    \label{point:cosmoi:quasi-representation-1}
  \item For each $\Mor[\mu]{n}{m}$ and $\Mor[\nu]{o}{n}$, there is a chosen lifting structure
    $F(\mu)(m) \pitchfork F(\mu\circ\nu)(\TY{o}) \times \El{m}$, where
    $\Mor[m]{F(\nu)(\EL{o})}{F(\nu)(\TY{o}) \times_{\TY{n}} \EL{n}}$ is the comparison map induced
    by \cref{diag:models:modal-intro}.
    \label{point:cosmoi:quasi-representation-2}
  \item $\El{\mu}$ contains a subuniverse also closed under all these connectives.
  \end{enumerate}
\end{definition}

To a first approximation, an \MTT{} cosmos is a standard model of \MTT{} formulated in the language
in natural models without the requirements that the universes be legitimate representable natural
transformations. By expanding the class of models in this way, we make it far easier to construct
the glued model in \cref{sec:normalization-model}.

\begin{remark}
  In order to simplify matters, we assume \MTT{} has only a weak universe \`a la Tarski. Concretely,
  there exists a universe $\Uni$ and a decoding operation $\Dec{-}$ just as in
  \textcite{gratzer:journal:2020}. Unlike the original presentation of \MTT{} however, we only
  obtain a chosen isomorphism between each connective and its decoded code \eg{}
  $\DecIso{-} : \Dec{\FnCode{A}{B}} \cong \Fn{\Dec{A}}{\Dec{B}}$.

  Recent unpublished work by Gratzer and Sterling has shown that Synthetic Tait
  Computability~\parencite{sterling:phd} can be strengthened to accommodate stricter universes,
  and can be applied to the following without issue. We remark, however, that there is experimental
  evidence that elaboration can be used to alleviate the tedium of these weaker universes in
  practice~\parencite{cooltt}.
\end{remark}

\begin{definition}
  \label{def:models:morphism-of-structured-cosmoi}
  A morphism between \MTT{}-structured cosmoi $\Mor[\alpha]{F}{G}$ is a 2-natural transformation
  $\alpha$ such that $\alpha_m$ is an LCCC functor and preserves all connectives strictly.
  Furthermore, we require that there is a natural isomorphism
  $\beta_\mu : \alpha_n \circ F(\mu)_! \cong G(\mu)_! \circ \alpha_m$ such that the transposition of
  $\alpha_\mu \circ \alpha_m(a)$ is $\alpha_m(\Transpose{a}) \circ \beta_\mu$ when
  $\Mor[a]{X}{F(\mu)(Y)} : F(m)$.
\end{definition}

\begin{theorem}
  Any strict model of \MTT{} induces a cosmos. In particular, the initial model of \MTT{}
  $\InterpSyn{-}$ induces a cosmos.
\end{theorem}

\begin{theorem}[Quasi-initiality]
  \label{thm:cosmoi:quasi-initiality}
  Given an arbitrary \MTT{} cosmos $G$ and a map $\Mor[\pi]{G}{\InterpSyn{-}}$, the following holds:
  \begin{enumerate}
  \item For every context $\IsCx{\Gamma}$, there exists an object $\Interp{\Gamma} : G(m)$, together
    with a canonical isomorphism $\alpha_\Gamma : \pi(\Interp{\Gamma}) \cong \Yo{\Gamma}$.
  \item For every type $\IsTy{A}$, there is a morphism $\Mor[\Interp{A}]{\Interp{\Gamma}}{\TY{m}}$
    such that $\pi(\Interp{A}) \circ \alpha_\Gamma = \YoEm{A}$.
  \item For every term $\IsTm{M}{A}$, there is a morphism $\Mor[\Interp{M}]{\InterpGl{M}}{\EL{m}}$
    lying over $\Interp{A}$ such that $\pi(\Interp{M}) \circ \alpha_\Gamma = \YoEm{M}$.
  \end{enumerate}
\end{theorem}
\begin{remark}
  We have used `syntactic' notation ($\IsCx{\Gamma}$, $\IsTm{M}{A{}}$, \etc{}) to denote the
  components of the initial model $\InterpSyn{-}$. This is a mere notational convenience: we do not
  rely on the fact that $\InterpSyn{-}$ may be constructed out of `traditional' syntax.
\end{remark}
\begin{proof}
  For clarity, we write $\SEl$, $\STy$ and $\STm$ instead of $\El{m}$, $\TY{m}$, and $\El{m}$ in the
  syntactic model, reserving the latter exclusively for $G$. We write $\Interp{\mu}$ for the functor
  sending $\Gamma$ to $\LockCx{\Gamma}$. We may replace $G$ by an equivalent strict 2-functor such
  that $\pi$ becomes strictly 2-natural, so we work under this assumption.

  We construct a displayed model of \MTT{}~\parencite{kaposi:qiits:2019} which lies over the
  syntactic model.
  \begin{itemize}
  \item Contexts in mode $m$ are interpreted by triples
    $(X : G(m), \Gamma, \alpha : \pi(X) \cong \Yo{\Gamma})$.
  \item A type in a context $(X, \Gamma, \alpha)$ is a pair of
    $(\Mor[\bar{A}]{X}{\TY{m}}, \IsTy{A})$ such that $\pi(\bar{A}) = \YoEm{A} \circ \alpha$.
  \item A term in a context $(X, \Gamma, \alpha)$ of type $(A^*, A)$ is a pair of
    $\Mor[M^*]{X}{\EL{m}}$ and $\IsTm{M}{A}$ such that $\El{m} \circ M^* = A^*$ and
    $\pi(M^*) = \YoEm{M} \circ \alpha$.
  \item A substitution between $(X, \Gamma, \alpha)$ and $(Y, \Delta, \beta)$ is a pair of
    $\Mor[f]{X}{Y}$ and $\IsSb{\delta}{\Delta}$ such that
    $\beta \circ \pi(f) = \Yo{\delta} \circ \alpha$
  \end{itemize}
  Once this model is constructed, the result is a direct consequence of the initiality of (strict)
  syntax. In fact, the construction of contexts, substitutions, terms, and types is relatively
  routine owing to the fact that $\pi$ is a 2-natural transformation, preserves finite limits,
  and strictly commutes with all connectives. We show a few cases to give a flavor for the
  procedure.

  \paragraph{The action of a modality on a context}
  Given a triple $(X, \Gamma, \alpha)$ at mode $n$ and a modality $\Mor[\mu]{n}{m}$, we define the
  `locked' context to be the following:
  \[
    (\LKan{G(\mu)}(X), \LockCx{\Gamma}, \gamma \circ \LKan{\Interp{\mu}}{\alpha} \circ \beta)
  \]
  Here, $\beta : \pi(\LKan{G(\mu)}(X)) \cong \LKan{\Interp{\mu}}(\pi(X))$ while
  $\gamma : \LKan{\Interp{\mu}}(\Yo{\Gamma}) \cong \Yo{\LockCx{\Gamma}}$.

  \paragraph{Context extension}
  Given a context $(X, \Gamma, \alpha)$ at mode $m$, a modality $\Mor[\mu]{n}{m}$ and a type
  $(A^*, A)$ in the context $(\LKan{G(\mu)}(X), \LockCx{\Gamma}, \beta)$, we form the context
  extension as the triple $(X \times_{G(\mu)(\TY{m})} G(\EL{m}), \ECx{\Gamma}{A}, \beta)$, where
  $\beta$ is the composite:
  \[
    \pi(X \times_{G(\mu)(\TY{m})} G(\mu)(\EL{m})) \cong
    \pi(X) \times_{\Pre{\Interp{\mu}}(\STy)} \Interp{\mu}^*(\STm) \cong
    \Yo{\ECx{\Gamma}{A}}
  \]
  The first isomorphism follows from the fact that $\pi$ preserves finite limits, is strictly
  2-natural, and strictly preserves $\El{m}$ while the second isomorphism is the universal property
  of $\Yo{\ECx{\Gamma}{A}}$.

  \paragraph{Modal types}
  Suppose we are given a context $(X, \Gamma, \alpha)$ and a type $(A^*, A)$ in the context
  $(\LKan{G(\mu)}(\mu)(X), \LockCx{\Gamma}, \gamma \circ \Interp{\mu}^*(\alpha) \circ \beta)$. We
  form the modal type as
  \[
    (\mathbf{Mod}_\mu(\Transpose{A^*}), \Modify{A})
  \]
  It remains to check that these types are coherent. That is, that
  \[
    \pi(\mathbf{Mod}_\mu(\Transpose{A^*})) = \YoEm{\Modify{A}} \circ \alpha
  \]
  By assumption, $\pi(A^*) = \YoEm{A} \circ \gamma \circ \Interp{\mu}^*(\alpha) \circ \beta$. By our
  assumption that $\pi$ commutes with transposition,
  $\pi(\Transpose{A^*}) = \Transpose{\YoEm{A} \circ \gamma} \circ \alpha$. The result follows from
  the fact that $\pi$ preserves $\mathbf{Mod}$.
\end{proof}

\subsection{An initial cosmos}

While the following material is not strictly necessary for our proof of normalization, we record it
for general interest.

\begin{theorem}
  \label{thm:cosmoi:sketchable}
  The category of MTT cosmoi can be sketched over $\VV$.
\end{theorem}
\begin{proof}
  First we show that $\VV$ is 2-monadic over $\Hom{\verts{\Mode}}{\CAT_g}$. This follows directly
  from the fact that locally Cartesian closure can be realized as a finitary 2-monad on
  $\CAT_g$~\parencite[Section 5.9]{lack:2009} and pseudofunctors whose 1-cells are right adjoints
  are finitarily 2-monadic over $\Hom{\verts{\Mode}}{\CAT_g}$~\parencite[Section
  5.9]{lack:2009}. Combining the operations and equations shows $\VV$ is finitarily 2-monadic over
  $\Hom{\verts{\Mode}}{\CAT_g}$.

  The constants and equations of \MTT{}-cosmoi are sketchable over this 2-monad following the work
  of \textcite{kinoshita-power-takeyama:1999}.
\end{proof}

\begin{corollary}
  \label{cor:cosmoi:initial}
  The category of \MTT{}-cosmoi has a bi-initial object $\ICosmos{-}$.
\end{corollary}


\section{Foundations of multimodal Synthetic Tait Computability}
\label{sec:foundations}

Eventually, we will wish to work with a collection of models of Synthetic Tait Computability (STC),
connected by a variety of adjunctions. In order for this to work, however, we need to show that
adjunctions between various categories lift to adjunctions between glued categories. The main
result of this section is \cref{thm:foundations:existence}, which states that a collection of glued
categories interconnected by right adjoints supports a model of ``multimodal'' STC.

\begin{theorem}
  \label{thm:foundations:induced}
  Suppose we are given lex functors $\Mor[\rho_i]{\SH{\XTop_i}}{\SH{\YTop_i}}$ and a 2-cell $\alpha$
  witnessing the weak commutativity of the following diagram:
  \[
    \DiagramSquare{
      width = 3cm,
      nw = \SH{\XTop_0},
      ne = \SH{\YTop_0},
      sw = \SH{\XTop_1},
      se = \SH{\YTop_1},
      north = \rho_0,
      south = \rho_1,
      east = \Inv{g},
      west = \Inv{f},
    }
  \]
  Then there is an induced morphism $\Mor[\brackets{f, g}]{\GL{\rho_1}}{\GL{\rho_0}}$ such
  that $\brackets{f,g}$ fits into the following diagram:
  \begin{equation*}
    \begin{tikzpicture}[diagram]
      \node (G1) {$\GL{\rho_1}$};
      \node [left = 3cm of G1] (X1) {$\XTop_1$};
      \node [right = 3cm of G1] (Y1) {$\YTop_1$};
      \node [below = of G1] (G0) {$\GL{\rho_0}$};
      \node [below = of X1] (X0) {$\XTop_0$};
      \node [below = of Y1] (Y0) {$\YTop_0$};
      \path[->] (G1) edge (G0);
      \path[->] (X1) edge (X0);
      \path[->] (Y1) edge (Y0);
      \path[open immersion] (X1) edge (G1);
      \path[embedding] (Y1) edge (G1);
      \path[open immersion] (X0) edge (G0);
      \path[embedding] (Y0) edge (G0);
    \end{tikzpicture}
  \end{equation*}
  Moreover, if both fringe functors $\rho_i$ are continuous and $f$ and $g$ are essential, then
  $\brackets{f,g}$ is also essential.
\end{theorem}
\begin{proof}
  Explicitly, $\Inv{\brackets{f,g}}$ sends $\Mor{F}{\rho_0(E)}$ to
  $\Mor{\Inv{g}(F)}{\rho_1(\Inv{f}(E))}$. The 2-cell $\alpha$ is used to obtain this morphism by
  correcting $g\parens{\Mor{F}{\rho_0(E)}}$.

  Colimits and finite limits are determined pointwise in $\SH{\GL{\rho_i}}$, so
  $\Inv{\brackets{f,g}}$ preserves them because $\Inv{f}$ and $\Inv{g}$ both do. If both $\rho_i$
  are continuous, then all limits are determined pointwise, so $\brackets{f,g}$ is essential if both
  $f$ and $g$ are.
\end{proof}

\begin{theorem}
  \label{thm:foundations:essential}
  Suppose that $f$ and $g$ are essential, then $\Ess{\brackets{f,g}}$ is a pointwise application of
  $\Ess{f}$ and $\Ess{g}$.
\end{theorem}
\begin{proof}
  Explicitly, fix a morphism $\Mor[x]{Y}{\rho_0(X)}$. We then define $\Ess{\brackets{f,g}}{x}$ to be
  $\Mor[\tilde{x}]{\Ess{g}(Y)}{\rho_1(\Ess{g}(X))}$, where $\tilde{x} = \beta \circ \Ess{g}(x)$ and
  $\beta$ is the canonical natural transformation determined by the following series of transposes:
  \begin{align*}
    \beta \in \Hom{\Ess{g} \circ \rho_1}{\rho_0 \circ \Ess{f}}
    &\cong \Hom{\rho_1}{\Inv{g} \circ \rho_1 \circ \Ess{f}}
    \\
    &\cong \Hom{\rho_1}{\rho_1 \circ \Inv{f} \circ \Ess{f}}
    \ni \rho_1(\Unit)
  \end{align*}

  It remains to show that this defines an adjoint. We show that there is a bijection
  $\Hom{\Ess{\brackets{f,g}}(x)}{x'} \cong \Hom{x}{\Inv{\brackets{f,g}}(x')}$. Given that
  $\Ess{\brackets{f,g}}$ and $\Inv{\brackets{f,g}}$ are defined pointwise, it suffices to show that
  the transpose operators for $\Ess{f} \Adjoint \Inv{f}$ and $\Ess{g} \Adjoint \Inv{g}$ lift.

  Fix $\Mor[b]{X}{\Inv{f}{X'}}$ and $\Mor[a]{Y}{\Inv{g}(Y')}$ such that
  $\alpha \circ \Inv{g}(x') \circ a = \rho_1(b) \circ x$. We wish to show that
  $x' \circ \Transpose{a} = \rho_0(\Transpose{b}) \circ \tilde{x}$.
  \begin{align*}
    \rho_0(\Transpose{b}) \circ \tilde{x}
    &= \rho_0(\Transpose{b}) \circ \beta \circ \Ess{g}(x)
    \\
    &= \rho_0(\Transpose{b}) \circ \Counit \circ \Ess{g}(\alpha^{-1} \circ \rho_1(\eta)) \circ \Ess{g}(x)
    \\
    &= \rho_0(\Transpose{b}) \circ \Counit \circ \Ess{g}(\alpha^{-1} \circ \rho_1(\eta) \circ x)
    \\
    &= \Counit \circ \Ess{g}(\Inv{g}(\rho_0(\Transpose{b})) \circ \alpha^{-1} \circ \rho_1(\eta) \circ x)
    \\
    &= \Counit \circ \Ess{g}(\alpha^{-1} \circ \rho_1(\Inv{f}(\Transpose{b})) \circ \rho_1(\eta) \circ x)
    \\
    &= \Counit \circ \Ess{g}(\alpha^{-1} \circ \rho_1(b) \circ x)
    \\
    &= \Counit \circ \Ess{g}(\Inv{g}(x') \circ a)
    \\
    &= x' \circ \Transpose{a}
  \end{align*}

  Next, $\Mor[b]{\Ess{f}(X)}{X'}$ and $\Mor[a]{\Ess{g}(Y)}{Y'}$ such that
  $x' \circ a = \rho_0(b) \circ \tilde{x}$.  We wish to show that
  $\alpha \circ \Inv{g}(x') \circ \Transpose{a} = \rho_1(\Transpose{b}) \circ x$.
  \begin{align*}
    \alpha \circ \Inv{g}(x') \circ \Transpose{a}
    &= \alpha \circ \Inv{g}(x') \circ \Inv{g}(a) \circ \Unit
    \\
    &= \alpha \circ \Inv{g}(x' \circ a) \circ \Unit
    \\
    &= \alpha \circ \Inv{g}(\rho_0(b) \circ \tilde{x}) \circ \Unit
    \\
    &= \rho_1(\Inv{f}(b)) \circ \alpha \circ \Inv{g}(\Counit \circ \Ess{g}(\alpha^{-1} \circ \rho_1(\eta) \circ x)) \circ \Unit
    \\
    &= \rho_1(\Inv{f}(b)) \circ \alpha \circ \Inv{g}(\Counit) \circ \Unit \circ \alpha^{-1} \circ \rho_1(\eta) \circ x
    \\
    &= \rho_1(\Inv{f}(b)) \circ \rho_1(\eta) \circ x
    \\
    &= \rho_1(\Transpose{b}) \circ x \qedhere
  \end{align*}
\end{proof}

It is natural to view $\GL{\rho_i}$ as a topos with a subterminal object $\Prop_i$ such that the
open subtopos corresponding to $\Prop_i$ is equivalent to $\XTop_i$, while the closed subtopos is
equivalent to $\YTop_i$. For the sake of easy of readability, we will identify these subtopoi with
$\XTop_i$ and $\YTop_i$ in what follows. Recall that the inverse image of the inclusion
$\EmbMor{\XTop_i}{\GL{\rho_i}}$ is given by $-^{\Prop_i}$. The inverse image of
$\EmbMor{\YTop_i}{\GL{\rho_i}}$ is given by $- \Join \Prop_i$, where $A \Join \Prop_i$ is the
pushout $A \Pushout{A \times \Prop_i} \Prop_i$.

\subsection{Open and closed subtopoi}
\label{sec:foundations:open-closed}

We now devote substantial effort to showing that the decomposition of $\GL{\rho_i}$ is preserved by
$\brackets{f,g}$.

\begin{lemma}
  If $A : \GL{\rho_0}$ is contained in $\SH{\XTop_0}$, then $\brackets{f,g}(A) : \SH{\XTop_1}$.
\end{lemma}
\begin{proof}
  Observe that $A$ is in the open subtopos $\SH{\XTop_0}$ if $A \cong \Mor{\rho_0(E)}{\rho_0(E)}$,
  in which case $\brackets{f,g}(A) \cong \Mor{\rho_1(f(E))}{\rho_1(f(E))}$ which then lies in the
  open subtopos.
\end{proof}

\begin{lemma}
  If $A : \SH{\GL{\rho_0}}$ is contained in $\SH{\YTop_0}$, then $\brackets{f,g}(A) : \SH{\YTop_1}$.
\end{lemma}
\begin{proof}
  Again, observe that $A : \SH{\YTop_0}$ if $A \cong \Mor{F}{\rho_0(\ObjTerm{\SH{\XTop_0}})}$, in which
  case
  \[
    \Inv{\brackets{f,g}}(A)
    \cong \Mor{g(F)}{\rho_1(f(\ObjTerm{\SH{\XTop_0}}))}
    \cong \Mor{g(F)}{\rho_1(\ObjTerm{\SH{\XTop_1}})}
  \]
  This then lies in the closed subtopos of $\GL{\rho_1}$.
\end{proof}

In fact, a more general theorem is true. Reflecting an object into either subtopos and applying
$\brackets{f,g}$ is (naturally) isomorphic to applying $\brackets{f,g}$ and then reflecting.
\begin{lemma}
  There is an isomorphism $\Mor{\Prop_1}{\brackets{f,g}(\Prop_0)}$.
\end{lemma}
\begin{proof}
  Unfolding definitions:
  \begin{align*}
    \Prop_0 &= \Mor{\ObjInit{\SH{\YTop_0}}}{\rho_0(\ObjTerm{\SH{\XTop_0}})}\\
    \Prop_1 &= \Mor{\ObjInit{\SH{\YTop_1}}}{\rho_1(\ObjTerm{\SH{\XTop_1}})}
  \end{align*}
  The result is immediate by computation.
\end{proof}

\begin{corollary}
  \label{cor:foundations:preserves-open-closed}
  \leavevmode
  \begin{enumerate}
  \item $\brackets{f,g}(\Prop_0) \cong \Prop_1$
  \item $\brackets{f,g}\parens{-^{\Prop_0}} \cong \parens{\brackets{f,g}\parens{-}}^{\Prop_1}$
  \item $\brackets{f,g}(\Prop_0 \Join{} -) \cong \Prop_1 \Join{} \brackets{f,g}\parens{-}$
  \end{enumerate}
\end{corollary}

\begin{remark}
  The last fact of \cref{cor:foundations:preserves-open-closed} is a consequence of the first;
  $\Join$ is preserved by $\brackets{f,g}$ because as this functor is lex and cocontinuous and
  $\Join$ is defined by finite products and pushouts. Exponentiation, however, is not generally
  preserved by $\brackets{f,g}$. It is a special fact of $\Prop_0$ that exponentiation by $\Prop_0$
  is preserved.
\end{remark}

As idempotent lex modalities, both of these actions internalize into the internal type theory of the
topoi. In particular, they have actions on a family given by applying the operation to the entire
family, then pulling back along the unit. This is automatically preserved by $\Inv{\brackets{f,g}}$,
as $\Inv{\brackets{f,g}}$ preserves each step in this construction. We write $\Open$ for the action
of $-^{\Prop_0}$ on families, and $\Closed$ for $- \Join \Prop_0$. We will abusively use the same
notation for $-^{\Prop_0}$ and $-^{\Prop_1}$ as well as $- \Join \Prop_0$ and $- \Join \Prop_1$.

\begin{theorem}
  Both $\Open$ and $\Closed$ are preserved by $\Inv{\brackets{f,g}}$
\end{theorem}

We would like to show a slightly stronger result, namely that the dependent versions of these
modalities are preserved~\parencite{rijke:2020}. Prior to this, we must address the
question of universes and models of type theory in this situation.

As a morphism of logoi, we know that $\brackets{f,g}$ is a cocontinuous left exact functor. We would
like to show that it induces a dependent right adjoint~\parencite{birkedal:2020}, but this is
complicated by the murky definition of a model of dependent type theory with universes in an
arbitrary logos. Accordingly, prove a theorem which is sufficient for the case we have in mind.

\begin{definition}
  A functor $\Mor[F]{\PSH{\CC}}{\PSH{\DD}}$ between presheaf categories is said to preserve a
  Grothendieck universe $\UU$ when it sends $\UU$-small families in $\PSH{\CC}$ to $\UU$-small
  families in $\PSH{\DD}$.
\end{definition}

\begin{theorem}
  A right adjoint $\Mor[F]{\PSH{\CC}}{\PSH{\DD}}$ which preserves Grothendieck universes larger
  than $\verts{\CC}$ and $\verts{\DD}$ induces a dependent right adjoint.
\end{theorem}
\begin{remark}
  To be precise, we must specify the particular models of dependent type theory we consider in
  $\PSH{\CC}$ and $\PSH{\DD}$. We take the model of type theory which a type to be an element of the
  universe of small presheaves~\parencite{hofmann-streicher:1997}.
\end{remark}
\begin{proof}
  Fix a Grothendieck universe $\VV$ which is large enough to contain both $\CC$ and $\DD$. We wish
  to show that $F$ induces a weak CwF morphism. A type in $\PSH{\CC}$ is given by a generalized
  element of $\TY{\CC}$, the Hofmann-Streicher universe of $\VV$-small presheaves, so it suffices to
  show that there is a pullback square of the following shape:
  \[
    \DiagramSquare{
      nw = F(\EL{\CC}),
      nw/style = pullback,
      west = F(\El{\CC}),
      ne = \EL{\DD},
      sw = F(\TY{\CC}),
      se = \TY{\DD},
      east = \El{\DD},
    }
  \]
  Proving this is equivalent to showing that $F(\El{\CC})$ is $\UU$-small because $\El{\DD}$ is
  generic for such maps. This immediately defines a weak CwF morphism using the observation that
  context extension in these models is defined by pullback.
\end{proof}

\begin{theorem}
  A cocontinuous lex functor $\Mor[F]{\PSH{\CC}}{\PSH{\DD}}$ preserves a Grothendieck universe $\UU$
  if and only if it restricts to a functor $\Mor{\PSH[\UU]{\CC}}{\PSH[\UU]{\DD}}$.
\end{theorem}
\begin{proof}
  The only if direction is clear, so it remains to show that preserving $\UU$-objects ensures that
  $F$ preserves $\UU$-families. Fix a $\UU$-small family $\Mor[f]{X}{Y} : \PSH{\CC}$.

  First, we observe that $F(Y) = \Colim_i F(\Yo{C_i})$ using the canonical decomposition of $Y$
  into colimits. By Yoneda, a morphism $\Mor{\Yo{D}}{F(Y)}$ then factors through some $F(\Yo{C_i}$),
  followed by $\Mor[F(\YoEm{y_i})]{F(\Yo{C_i})}{F(Y)}$, so it suffices to show that the pullback of
  $F(f)$ to $F(\Yo{C_i})$ is a small family.

  As $F(f)$ and $F(\YoEm{y_i})$ are both in the image of $F$, this pullback can be computed in
  $\PSH{\CC}$. Therefore, we must show that $\Mor{F(X \times_{Y} \Yo{C_i})}{F(\Yo{C_i})}$ is a small
  family. Using the assumption that $f$ was a small family, we observe that $X \times_Y \Yo{C_i}$ is
  a small object, so $F(X \times_Y \Yo{C_i})$ is $\UU$-small. Therefore, this family is a small
  family, completing the proof.
\end{proof}

While these results are limited to presheaf topoi, they are sufficient for our purposes; we are
interested in gluing together presheaf topoi along a continuous functor. The following result shows
that the resultant topos is of presheaf type:

\begin{theorem}[\textcite{sga:4,carboni:1995}]
  Gluing together presheaf topoi along a continuous functor results in a presheaf topos.
\end{theorem}

Motivated by this, we now assume that $\EE_i$ and $\FF_i$ are all presheaf logoi, and all morphisms
involved preserve all Grothendieck universes larger than $\UU$. In this case, we may then pick a
Grothendieck universe $\UU$ large enough to be preserved by both $f$ and $g$ and so the bases of
$\GL{\rho_1}$ and $\GL{\rho_1}$ are small for this universe. In this case, $\brackets{f,g}$ is a DRA
which preserves all universes larger than $\UU$.

We now show that $\brackets{f,g}$ preserves the dependent version of $\Open$ and
$\Closed$.

\begin{theorem}
  $\Inv{\brackets{f,g}}$ preserves $\Dep{\Open}$ and $\Dep{\Closed}$ on any universe preserved by
  $\Inv{\brackets{f,g}}$.
\end{theorem}
\begin{proof}
  We will show this for $\Dep{\Open}$, though the proof is identical for $\Dep{\Closed}$ and indeed
  for any lex modality preserved by $\brackets{f,g}$.

  Let us pick some universe $\UU$ large enough to be preserved by $\brackets{f,g}$ and write
  $\Mor[\El{0}]{\EL{0}}{\TY{0}}$ for the Hofmann-Streicher universe it induces in $\SH{\GL{\rho_0}}$
  and $\Mor[\El{1}]{\EL{1}}{\TY{1}}$ for the Hofmann-Streicher universe in $\SH{\GL{\rho_1}}$. By
  assumption, we have a pullback squares
  \[
    \DiagramSquare{
      nw = \Inv{\brackets{f,g}} \EL{0},
      nw/style = pullback,
      sw = \Inv{\brackets{f,g}} \TY{0},
      west = \Inv{\brackets{f,g}} \El{0},
      ne = \EL{1},
      se = \TY{1},
      east = \El{1},
      south = i,
    }
    \qquad
    \DiagramSquare{
      nw = \Open \EL{i},
      nw/style = pullback,
      sw = \Open \TY{i},
      west = \Open \El{i},
      ne = \EL{i},
      se = \TY{i},
      east = \El{i},
      south = \Dep{\Open},
    }
  \]
  Finally, we note that $\Inv{\brackets{f,g}}$ preserves $\Open$, meaning that there is a natural
  isomorphism $\Mor[\alpha]{\Inv{\brackets{f,g}} \circ \Open}{\Open \circ \Inv{\brackets{f,g}}}$.

  We may now rephrase: our goal is to show that $i \circ \Inv{\brackets{f,g}}(\Dep{\Open})$ and
  $\Dep{\Open} \circ \Open i \circ \alpha$ represent the same type. First, compute the type
  represented by $\Dep{\Open} \circ \Open i \circ \alpha$:
  \[
    \begin{tikzpicture}[diagram]
      \node[pullback] (F-modal-tm) {$\Inv{\brackets{f,g}} \Open \EL{0}$};
      \node[pullback, right = 3cm of F-modal-tm] (modal-tm) {$\Open \EL{1}$};
      \node[right = of modal-tm] (tm) {$\EL{1}$};
      \node[below = of F-modal-tm] (F-modal-ty) {$\Inv{\brackets{f,g}} \Open \TY{0}$};
      \node[below = of modal-tm] (modal-ty) {$\Open \TY{1}$};
      \node[below = of tm] (ty) {$\TY{1}$};

      \path[->] (F-modal-tm) edge (modal-tm);
      \path[->] (modal-tm) edge (tm);

      \path[->] (F-modal-ty) edge node[below] {$\Open i \circ \alpha$} (modal-ty);
      \path[->] (modal-ty) edge node[below] {$\Dep{\Open}$} (ty);

      \path[->] (F-modal-tm) edge (F-modal-ty);
      \path[->] (modal-tm) edge (modal-ty);
      \path[->] (tm) edge (ty);
    \end{tikzpicture}
  \]
  The left-hand square is a pullback because $\Open$ is lex. Next, we compute
  $i \circ \Inv{\brackets{f,g}}(\Dep{\Open})$:
  \[
    \begin{tikzpicture}[diagram]
      \node[pullback] (F-modal-tm) {$\Inv{\brackets{f,g}} \Open \EL{0}$};
      \node[pullback, right = 3cm of F-modal-tm] (F-tm) {$\Inv{\brackets{f,g}} \EL{1}$};
      \node[right = of F-tm] (tm) {$\EL{1}$};
      \node[below = of F-modal-tm] (F-modal-ty) {$\Inv{\brackets{f,g}} \Open \TY{0}$};
      \node[below = of F-tm] (F-ty) {$\Inv{\brackets{f,g}} \TY{1}$};
      \node[below = of tm] (ty) {$\TY{1}$};

      \path[->] (F-modal-tm) edge (F-tm);
      \path[->] (F-tm) edge (tm);

      \path[->] (F-modal-ty) edge node[below] {$\Inv{\brackets{f,g}} \Dep{\Open}$} (F-ty);
      \path[->] (F-ty) edge node[below] {$i$} (ty);

      \path[->] (F-modal-tm) edge (F-modal-ty);
      \path[->] (modal-tm) edge (modal-ty);
      \path[->] (tm) edge (ty);
    \end{tikzpicture}
  \]
  The left-hand square is now a pullback because $\Inv{\brackets{f,g}}$ is lex.
\end{proof}

\subsection{Interpreting \MTT{} in glued topoi}

The results of the previous subsection show that $\brackets{f,g}$ preserves the open and closed
modalities of a glued topos. As a dependent right adjoint, moreover, $\brackets{f,g}$ can be
internalized into a model of \MTT{}. It is this model of \MTT{} that we use to substantiate the
language of `multimodal' synthetic Tait computability.

\begin{theorem}
  \label{thm:foundations:existence}
  Fix a 2-category $\Mode$ and a 2-functor $\Mor[\Interp{-}]{\Op{\Mode} \times \SIMP^1}{\CAT}$ which
  assigns 0-cells to presheaf topoi and 1-cells to essential geometric morphisms between them. There
  exists a model of \MTT{} with $\Mode$ such that mode $m$ is given by
  $\GL{\Inv{\Interp{\Mor{m \times 0}{m \times 1}}}}$ and the modality $\mu$ is interpreted by
  $\brackets{\Inv{\Interp{\mu \times 0}}, \Inv{\Interp{\mu \times 1}}}$.

  Furthermore, each mode has an open and closed modality preserved by the \MTT{} modalities such
  that the open (resp. closed) subtopos of mode $m$ is equivalent to $\Interp{m \times 0}$,
  (resp. $\Interp{m \times 1}$).
\end{theorem}

In this model, each mode supports a model of synthetic Tait computability and the extra structure of
this language (the open and closed modalities) is preserved by the modalities between modes.


\section{The category of renamings for \MTT{}}
\label{sec:renamings}

We now isolate a class of renamings: $\Ren{m}$. We also define the neutral and normal forms of terms
in \MTT{}. Unlike the syntax of \MTT{}, neutral and normal forms are not taken up to a complex
equivalence relation, and so their conversion problem is immediately reducible to conversion in the
mode theory.

The judgments for these are slightly atypical, in that the typing judgments for both neutral and
normal forms are defined (1) with formal telescopes as contexts and (2) inductive-recursively with
the inclusion of neutral and normals into terms, and renamings into contexts and substitutions.
This dependence of the unquotiented raw syntax on actual terms is necessary in order to ensure that
the typing relation can be sensibly defined.

The full collection of neutral and normal forms is given in
\cref{sec:appendix:neutrals-and-normals}. We present a few representative cases below.

\begin{mathparpagebreakable}
  \JdgFrame{\IsTele{\Theta}{\Gamma} \qquad \IsRen{\psi}{\Psi}{\gamma}}
  \\
  \inferrule{
    \IsTele{\Theta}{\Gamma}
  }{
    \IsTele{\LockTele{\Theta}}{\LockCx{\Gamma}}<n>
  }
  \and
  \inferrule{
    \IsTele{\Gamma, \Delta}{\Theta,\Psi}
    \\
    \IsRen{r}{\Psi}{\delta}
  }{
    \IsRen[\LockTele{\Theta}]{\LockRen{r}}{\LockTele{\Psi}}<n>{\LockSb{\delta}}
  }
  \and
  \inferrule{
    \IsTele{\Theta}{\Gamma}
    \\
    \Mor[\mu, \nu]{n}{m}
    \\
    \Mor[\alpha]{\nu}{\mu}
  }{
    \IsRen[\LockTele{\Theta}]{\KeyRen{\alpha}{\Theta}}{\LockTele{\Theta}<\nu>}<n>{\Key{\alpha}{\Gamma}}
  }
  \\
  \JdgFrame{\IsNfTy{\tau}{A}}
  \\
  \inferrule{
    \IsNfTy{\tau}<n>{A}
    \\
    \IsNfTy[\ETele{\Theta}{A}<\ArrId{m}>]{\sigma}{B}
  }{
    \IsNfTy{\NfProd{\tau}{\sigma}}{\Sig{A}{B}}
  }
  \and
  \inferrule{
    \IsNfTy[\LockTele{\Theta}]{\tau}<n>{A}
  }{
    \IsNfTy{\NfModify{\tau}}{\Modify{A}}
  }
  \and
  \inferrule{
    \IsTele{\Theta}{\Gamma}
    \and
    \IsNf{u}{\NfUni}{A}
  }{
    \IsNfTy{\NfDec{u}}{A}
  }
  \\
  \JdgFrame{
    \IsNe{e}{A}{M}
    \qquad
    \IsNf{u}{A}{M}
  }
  \\
  \inferrule{
    \IsTele{\Theta}{\Gamma}
    \\
    \Theta(k) = (\DeclNameless{A})
    \\
    \Locks{\Theta}{k} = \nu
    \\
    \Mor[\alpha]{\mu}{\nu}
  }{
    \IsNe{\NeVar{k}{\alpha}}{
      \Sb{A}{\Key{\alpha}{} \circ (\LockSb{\Wk}<\nu_{k-1}>) \dots \circ (\LockSb{\Wk}<\nu_{0}>)}
    }{
      \Sb{\Var{0}}{\Key{\alpha}{} \circ (\LockSb{\Wk}<\nu_{k-1}>) \dots \circ (\LockSb{\Wk}<\nu_{0}>)}
    }
  }
  \and
  \inferrule{
    \IsTele{\Theta}{\Gamma}
    \\
    \IsTy[\LockCx{\Gamma}]{A}<n>
    \\
    \IsTy[\ECx{\Gamma}{A}]{B}
    \\\\
    \IsNf[\ETele{\Theta}{A}]{u}{B}{M}
  }{
    \IsNf{\NfLam{u}}{\Fn{A}{B}}{\Lam{M}}
  }
  \and
  \inferrule{
    \IsTele{\Theta}{\Gamma}
    \\
    \IsTy[\LockCx{\Gamma}]{A}<n>
    \\
    \IsTy[\ECx{\Gamma}{A}]{B}
    \\\\
    \IsNe{e}{\Fn{A}{B}}{M}
    \\
    \IsNf{u}{A}{N}
  }{
    \IsNe{\NeApp{e}{u}}{\Sb{B}{\ESb{\ISb}{N}}}{\App{M}{N}}
  }
  \and
  \inferrule{
    \IsTele{\Theta}{\Gamma}
    \\
    \IsTy[\LockCx{\Gamma}]{A}<n>
    \\
    \IsNf[\LockTele{\Theta}]{u}{A}<n>{M}
  }{
    \IsNf{\NfMkBox{u}}{\Modify{A}}{\MkBox{M}}
  }
  \and
  \inferrule{
    \IsTele{\Theta}{\Gamma}
    \\
    \IsTy[\LockCx{\Gamma}]{A}<n>
    \\
    \IsNe{e}{\Modify{A}}{M}
  }{
    \IsNf{\NfInj{e}}{\Modify{A}}{M}
  }
  \and
  \inferrule{
    \Mor[\nu]{o}{n}
    \\
    \Mor[\mu]{n}{m}
    \\
    \IsTele{\Theta}{\Gamma}
    \\
    \IsTy[\LockCx{\LockCx{\Gamma}}<\nu>]{A}<o>
    \\
    \IsNe[\LockTele{\Theta}]{u}{\Modify[\nu]{A}}<n>{M}
    \\
    \IsNfTy[\ETele{\Theta}{\Modify[\nu]{A}}]{\tau}{B}
    \\
    \IsNf[\ETele{\Theta}{A}<\mu \circ \nu>]{u}{\Sb{B}{\ESb{\Wk}{\MkBox[\nu]{\Var{0}}}}}
  }{
    \IsNe{\NeLetMod{\mu}{\nu}{\tau}{e}{u}}{\Sb{B}{\ESb{\ISb}{M}}}{\LetMod{M}{N}}
  }
  \and
  \inferrule{
    \IsTele{\Theta}{\Gamma}
    \\
    \IsNe{e}{\Dec{\BoolCode}}{M}
  }{
    \IsNe{\NfDecIso{e}}{\Bool}{\DecIso{M}}
  }
  \and
  \inferrule{
    \IsTele{\Theta}{\Gamma}
    \\
    \IsNf{u}{\Bool}{M}
  }{
    \IsNf{\NfDecIso*{u}}{\Dec{\BoolCode}}{\DecIso*{M}}
  }
\end{mathparpagebreakable}

There are no equations imposed on any of these generators, except for renamings which have the
necessary equations to organize them into a 2-functor into $\CAT$. We further define a substitution
action applying a renaming to a normal or neutral form, which must show respects these
equations. This operation and its pertinent proofs are largely standard, with the exception being
the case for variables. We reproduce the definition of the substitution action on variables here:
\begin{align*}
  \Sb{\NeVar{k}{\alpha}}{\IRen} &= \NeVar{k}{\alpha}\\
  \Sb{\NeVar{k}{\alpha}}{\LockRen{\WkRen}} &= \NeVar{k + 1}{\alpha}\\
  \Sb{\NeVar{k}{\alpha}}{\LockRen{\ERen{r}{\NeVar{j}{\beta}}}} &=
  \begin{cases}
    \NeVar{j}{(\beta \Whisker \ArrId{\mu})\circ\alpha} & k = 0\\
    \Sb{\NeVar{k-1}{\alpha}}{\LockRen{r}} & \text{otherwise}
  \end{cases}\\
  \Sb{\NeVar{k}{\alpha}}{\LockRen{(r \circ s)}} &= \Sb{\Sb{\NeVar{k}{\alpha}}{\LockRen{r}}}{\LockRen{s}}\\
  \Sb{\NeVar{k}{\alpha}}{\KeyRen{\beta}{\Theta}} &= \NeVar{k}{\alpha\circ\beta}
\end{align*}

\begin{remark}
  Note that these normal forms do not necessarily enjoy decidable equality. Rather, the problem of
  deciding when two normal forms are convertible is precisely the problem of deciding whether
  certain 1- and 2-cells of the mode theory $\Mode$ are equal. Accordingly, it is possible that
  \MTT{} may enjoy normalization, but not decidable type-checking. The mode theories used in
  instantiations of \MTT{} thus far are evidently decidable, so we do not expect this to be an issue
  for practical usage.
\end{remark}


\section{Prerequisites for the normalization model}
\label{sec:prereq}

\subsection{Key syntactic categories}

Recall from \cref{sec:renamings} that for each mode there is a category $\Ren{m}$ of telescopes and
renamings, together with a functor $\Mor[\EmbRen]{\Ren{m}}{\Cx{m}}$ which sends a telescope the
corresponding context and formal renaming to substitution. The collection of these functors induces
a natural transformation between the 2-functors $\Mor[\Cx{-},\Ren{-}]{\Coop{\Mode}}{\CAT}$.

The functor $\EmbRen$ induces a functor $\Mor[\InvRen]{\PSH{\Cx{m}}}{\PSH{\Ren{m}}}$. Moreover,
because $\PSH{-}$ preserves strict equalities, this gives rise to a 2-functor
$\Mor{\Mode \times \SIMP^1}{\CAT}$. As a morphism defined by precomposition each $\InvRen$ has a
left adjoint and a right adjoint: $\ExtEmbRen \Adjoint \InvRen \Adjoint \Dir{\EmbRen}$. The left
adjoint extends $\EmbRen$ in the sense that $\ExtEmbRen{\Yo{\Theta}} \cong \Yo{\EmbRen(\Theta)}$. We
may now apply \cref{thm:foundations:existence}.

\begin{theorem}
  \label{thm:prereq:model}
  There exist a model $\InterpGl{-}$ of \MTT{} with a hierarchy of cumulative universes in
  $\GL{\InvRen}$.

  In this model, $\InterpGl{m}$ is $\GL{\InvRen}$, and $\InterpGl{\mu}$ is interpreted by
  precomposition with $\mu$ in both components. Furthermore, there is a (dependent) open and closed
  modality in $\InterpGl{m}$, preserved by all modalities.
\end{theorem}

In what follows, we will freely use \MTT{} to work with this model. As the interpretation of
identity types in this model supports equality reflection, we will work with extensional equality
when reasoning with \MTT{}.

\subsection{Basic properties of $\InterpGl{m}$}

The \MTT{} modalities in this model are particularly well-behaved because the arise from the inverse
image of essential geometric morphisms. In particular, they preserve colimits internally to the
theory.
\begin{theorem}
  \label{thm:prereq:internal-left-adj}
  For each $\mu$ there is an equivalence $\Modify{A + B} \Equiv \Modify{A} + \Modify{B}$.
\end{theorem}
\begin{proof}
  Recall that each modality in this model is interpreted by a morphism which is simultaneously a
  left and right adjoint. Therefore, this model can be extended to the mode theory $\AdjMode$, which
  extends $\Mode$ by adding a right adjoint $\bar{\mu}$ for each $\mu$ from
  $\Mode$~\parencite[Section 10.2]{gratzer:journal:2020}.

  The \emph{crisp induction} principles derived in \textcite[Section 10.5]{gratzer:journal:2020} now
  show that $\Modify{-}$ preserves coproducts internally with this extended mode theory, so the
  equivalence is validated by the model from \cref{thm:prereq:model} (though is no longer internally
  derivable).
\end{proof}

\begin{remark}
  We note that this model of \MTT{} enjoys
  $\ECx{\Gamma}{A}<\mu\circ\nu> \cong \ECx{\Gamma}{\Modify[\nu]{A}}$: the modalities are interpreted
  by proper dependent right adjoints. We therefore blur the distinction between a variable
  $\DeclVar{x}{A}$ and a variable $\DeclVar{x}{\Modify{A}}<\ArrId{}>$.
\end{remark}

\begin{notation}
  It's convenient to recall that $\Open$ is defined by exponentiation with a subterminal (a
  proposition) $\Prop$. We accordingly write $\Open[z : \Prop] A(z)$.
\end{notation}

\begin{remark}
  Strictly speaking, $\Prop$ should contain a mode annotation. However,
  $\Modify{\Prop_n} \cong \Prop_m$, and as both are subterminal we identify $\Prop_m$ and
  $\Modify{\Prop_n}$.
\end{remark}

Finally, we note also that $\Open$ and $\Closed$ enjoy two important properties:
\begin{theorem}[Fracture~\parencite{sga:4}]
  \label{thm:prereq:fracture}
  For any type $A$ we have $A \cong \Open A \times_{\Closed\Open A} \Closed A$.
\end{theorem}

\begin{theorem}[Internal realignment~\parencite{orton:2018,sterling:modules:2020}]
  \label{thm:prereq:realignment}
  Let us denote the type of types isomorphic to $A$ as $\Iso(A)$:
  \[
    \Iso(A) = \Sum{B : \Uni} \parens{A \cong B}
  \]
  There is a section $\Realign$ to $\Mor[\eta]{\Iso(A)}{\Open \Iso(A)}$.
\end{theorem}
\begin{remark}
  The following constructive proof is due to Christian Sattler.
\end{remark}
\begin{proof}
  This follows directly from the argument given by \textcite[Theorem 8.4]{orton:2018}. It suffices
  to show that $\Prop$ is levelwise decidable externally. This is immediate when we unfold the
  presentation of $\InterpGl{m}$ as a presheaf topos.

  Explicitly, \textcite{carboni:1995} show that $\InterpGl{m}$ can be presented as
  $\PSH{\COLLAGE{\Hom{-}{\EmbRen\parens{-}}}}$. While the precise definition of
  $\COLLAGE{\Hom{-}{\EmbRen\parens{-}}}$ is given by \textcite{carboni:1995}, for our purposes it
  suffices to recall that the collection of objects of this \emph{collage} is the disjoint union of
  the objects of $\Cx{m}$ and $\Ren{m}$. Under this presentation moreover, $\Prop$ can be defined as
  follows:
  \[
    \Prop(\Theta) = \emptyset \qquad \Prop(\Gamma) = \braces{\star}
  \]
  The levelwise decidability of $\Prop$ is then immediate.
\end{proof}

In anticipation of using this model of \MTT{} to construct the normalization model in the next
section, we also introduce several constants interpreted by specific presheaves in $\InterpGl{m}$.

\begin{definition}
  There is a constant $\IsTm[]{\Ty{m}}{\Open \Uni}$ in $\InterpGl{m}$ given by $\SynTy{m}$.
\end{definition}

\begin{definition}
  There is a type $z : \Prop, \Ty{m}(z) \vdash \Tm{m} : \Uni$ in $\InterpGl{m}$ whose total space is
  given by $\Mor[\SynEl{m}]{\SynTm{m}}{\SynTy{m}}$.
\end{definition}

\begin{definition}
  There is pair of types $\Open[z] \Ty{m}(z) \vdash \Nf{m}$ and $\Open[z] \Ty{m}(z) \vdash \Ne{m}$
  in $\InterpGl{m}$ given by normal and neutral forms viewed as a presheaf over
  $\InvRen{\SynTm{m}}$. Similarly, there is a closed type $\NfTy{m}$ of normal types lying over
  $\InvRen{\SynTy{m}}$.
\end{definition}

\begin{definition}
  There is a type $\Open[z] \Ty{m}(z) \vdash \Vars{m}$ given by viewing the presheaf of variables as
  lying over $\InvRen{\SynTm{m}}$. There is an induced inclusion $\EmbMor{\Vars{m}(A)}{\Ne{m}(A)}$
  which we shall treat as silent.
\end{definition}

Using realignment if necessary, we choose $\Nf{m}$, $\Ne{m}$, $\Vars{m}$ such that under the
assumptions $z : \Prop$ and $A : \Open[z] \Ty{m}(z)$, we have the following equations:
\[
  \Vars{m}(A) = \Ne{m}(A) = \Nf{m}(A) = \Tm{m}(z, A(z))
\]
Similarly, under just the assumption of $z$, we have $\NfTy{m} = \Ty{m}(z)$.

\subsection{An \MTT{} cosmos, internally}
\label{sec:prereq:internal-cosmos}

In \cref{sec:normalization-model} we will construct a model of \MTT{} in $\InterpGl{-}$ lying over
the syntactic model. To that end, we recast the existence of a model of \MTT{} into a structure we
can express in the internal language of $\InterpGl{-}$. Explicitly, we rephrase
\cref{def:cosmoi:cosmos} as a sequence of constants in the internal type theory \eg{}, a pair
$\Ty{m}$ and $\Ty{m} \vdash \Tm{m}$, combinators for all the type codes, \etc{} The full list of
constants is described in \cref{sec:appendix:cosmos}.

\begin{remark}
  There is potential for confusion here: we are using the internal language of $\InterpGl{-}$
  (extensional \MTT{}) as a framework to express the structure of an \MTT{} cosmos. This is not a
  circularity; the interpretation of extensional \MTT{} into $\InterpGl{-}$ is formulated in terms
  of the models and metatheory developed already in \textcite{gratzer:journal:2020}.
\end{remark}

Already, we can immediately obtain an \MTT{} cosmos internal to $\InterpGl{-}$ by restricting to the
open subtopos: $\InterpSyn{-}$. In this cosmos, types are interpreted as elements of the constant
$\Ty{m}(z)$ while terms are elements of $\Tm{m}(z, -)$. This the \emph{syntactic} \MTT{}
cosmos. Externally, this is the \MTT{} cosmos $\InterpSyn{-}$ where universes are interpreted by the
traditional representable natural transformations~\parencite{awodey:2018}. Eventually, we will
construct the normalization model as a second cosmos in $\InterpGl{-}$ which lies strictly over the
syntactic cosmos.

We will use the ``unqualified'' names such as $\PiConst(z)$ to refer to elements of this syntactic
cosmos. So, for instance, under the assumption $z : \Prop$ we the following:
\begin{align*}
  \ModConst &: \DeclVar{A}{\Ty{n}(z)} \to \Ty{m}(z)
  \\
  \ModIntroConst &: \DeclVar{A}{\Ty{n}(z)}<\mu>\DeclVar{a}{\Tm{n}(z, A)} \to \Tm{m}(\ModConst(A))
  \\
  \ModElimConst &:
  \DeclVar{A}{\Ty{n}(z)}<\nu\circ\mu>\brackets{\DelimMin{1}B : \Modify{\Tm{n}(z, \ModConst(A))} \to \Ty{o}(z)} \to\\
  &\qquad \brackets{\DelimMin{1}\DeclVar{x}{\Tm{n}(z, A)}<\nu\circ\mu> \to \Tm{o}(z, B(\ModIntroConst(A, x)))} \to\\
  &\qquad \DeclVar{a}{\Tm{m}(z, \ModConst(A))}<\nu> \to\\
  &\qquad \Tm{o}(z, B(a))\\
  \_ &:
  \DeclVar{A}{\Ty{n}(z)}<\nu\circ\mu>\brackets{\DelimMin{1}B : \Modify{\Tm{n}(z, \ModConst(A))} \to \Ty{o}(z)} \to\\
  &\qquad \brackets{\DelimMin{1}b : \DeclVar{x}{\Tm{n}(z, A)}<\nu\circ\mu> \to B(\ModIntroConst(A, x))} \to\\
  &\qquad \DeclVar{a}{\Tm{n}(z, A)}<\nu\circ\mu> \to \ModElimConst(A, B, b, \ModIntroConst(A, a)) = b(a)
\end{align*}

Now, because $\Open$ commutes with dependent products, sums, and modalities, we may equivalently
view, \eg{} $\PiConst$ as a constant with the following type:
\[
  \brackets{\Sum{A : \Modify{\Open[z] \Ty{n}(z)}} \Modify{\Open[z] \Tm{n}(z, A)} \to \Open[z] \Ty{m}(z)}
  \to \Open[z] \Ty{m}(z)
\]
Similarly, we will abusively write $\PiConst(A,B)$ for $\lambda z.\ \PiConst(A(z), B(-,
z))$.

\subsection{Higher-order abstract syntax for neutral and normal forms}
\label{sec:prereq:hoas}

The types $\Ne{m}$, $\Nf{m}$ can be used to encode a form of HOAS inside of
$\InterpGl{m}$~\parencite{hofmann:1999}. We record the constants that result from this in
\cref{sec:appendix:hoas}.

In the proof of normalization we will take advantage of the fact that under an open modality, a
normal and neutral form decode to the appropriate terms. Therefore, under the open modality
$\CProj{0}(\CPair(M,N))$ is well-typed and equal to $M$ along with other expected equations.


\section{The normalization model}
\label{sec:normalization-model}

We now construct a model of \MTT{} in $\InterpGl{-}$ which lies over the over the syntactic model of
\MTT{} built around $\Ty{m}$ and $\Tm{m}$ introduced in
\cref{sec:prereq:internal-cosmos}. Concretely, this means that we must construct a series of
constants ($\Ty*{m}$, $\Tm*{m}$, $\PiConst*$, \etc{}) in $\InterpGl{-}$ such that under $z : \Prop$
these constants are equal to their corresponding syntactic components ($\Ty{m}(z)$, $\Tm{m}(z, -)$,
$\PiConst(z)$, \etc{}). This family of constants will define a model in $\InterpGl{-}$ and the
`alignment' condition ensures that there is a morphism of models
(\cref{thm:normalization-model:model}). See \cref{sec:appendix:cosmos} for the full list of
constants. Throughout this section we will use extensional \MTT{} as an internal language for
$\InterpGl{-}$ (\cref{thm:prereq:model}) to define these constants. The strict equations will follow
from repeated applications of the realignment theorem (\cref{thm:prereq:realignment}). Many of
these computations will be familiar to readers experienced with STC. The main novelties are modal
connectives \cref{lem:normalization-model:pi,lem:normalization-model:mod}.

\begin{definition}
  \label{def:normalization-model:ty}
  We define $\Ty*{m}$ as the realignment of the following over $\Ty{m}$:
  \begin{equation*}
    \begin{make-rcd}{T}{\Ext{\Uni[2]}{z : \Prop}{\Ty{m}}}
      \Code{} : \NfTy{m}
      \\
      \Pred{} : \Ext{\Uni[1]}{z : \Prop}{\Tm{m}(z, \Code{})}
      \\
      \Reflect{} : \Ext{\Ne{m}(\Open \Code{}) \to \Pred{}}{z : \Prop}{\ArrId{}}
      \\
      \Reify{} : \Ext{\Pred{} \to \Nf{m}(\Open \Code{})}{z : \Prop}{\ArrId{}}
    \end{make-rcd}
  \end{equation*}
  Taking advantage of the isomorphism $\Ty*{m} \cong T$, we construct elements of $\Ty*{m}$ by
  specifying $\Code{}$, $\Pred{}$, $\Reflect{}$, and $\Reify{}$. By realignment,
  $\Unit(\braces{\Code{}; \Pred{}; \Reify{}; \Reflect{}}) = \Unit(\Code{})$.
\end{definition}

\begin{definition}
  \label{def:normalization-model:tm}
  We define $\IsTy[A : \Ty*{m}]{\Tm*{m}(A)}$ by $\Tm*{m}(\braces{\_;\Pred{};\_;\_}) = \Pred{}$.
\end{definition}

\begin{lemma}
  \label{lem:normalization-model:tm-align}
  $\Tm*{m}$ lies strictly over $\Tm{m}$. Explicitly for each $A : \Ty*{m} = \Ty{m}(z)$,
  $\Tm*{m}(A) = \Tm{m}(z, A)$.
\end{lemma}
\begin{proof}
  Using the boundary condition on $\Pred{}$, we compute
  $\Tm*{m}(A) = \Pred{A} = \Tm{m}(z, \Code{A})$. Under the hypothesis $z : \Prop$, we obtain an
  equality $A = \Code{A}$ because $\Unit(A) = \Unit(\Code{A})$ by definition. The conclusion now
  follows.
\end{proof}

\begin{lemma}
  \label{lem:normalization-model:sig}
  Fixing $T_0 : \Ty*{m}$ and $T_1 : \Tm*{m}(T_0) \to \Ty*{m}$, there exists a pair of constants:
  \begin{align*}
    \SigConst* &: \Ext{\Ty*{m}}{z : \Prop}{\SigConst(z, T_0, T_1)}\\
    \alpha_{\SigConst*} &:
    \Ext{\Tm{m}(\SigConst*(T_0, T_1)) \cong \Sum{t : \Tm*{m}(T_0)} \Tm*{m}(T_1(t))}{z : \Prop}{\alpha_{\SigConst}(z,T_0,T_1)}
  \end{align*}
\end{lemma}
\begin{proof}
  Let us start by apply \cref{thm:prereq:realignment} to $\alpha_{\SigConst(z)}$ and
  $\Sum{t_0 : \Pred{T_0}} \Pred{T_1(t_0)}$ to produce $\Psi : \Uni[1]$ such that
  $z : \Prop \vdash \Psi = \SigConst(z, T_0, T_1)$ and
  $\alpha_{\SigConst*} : \Psi \cong \Sum{t_0 : \Pred{T_0}} \Pred{T_1(t_0)}$ which restricts to
  $\alpha_{\SigConst}$ under $z : \Prop$.

  We now define $\SigConst*(T_0,T_1)$ as follows:
  \begin{align*}
    \Code{\SigConst*(T_0,T_1)} &= \CSig(\Code{T_0}, \lambda v.\ \Code{T_1(\Reflect{T_0} v)})
    \\
    \Pred{\SigConst*(T_0,T_1)} &= \Psi
    \\
    \Reflect{\SigConst*(T_0,T_1)} &=
    \lambda e.\ \alpha_{\SigConst*}^{-1} \gls{\Reflect{T_0}(\CProj{0}(e)), \Reflect{T_1(\Reflect{T_0}(\CProj{0}(e)))}(\CProj{1}(e))}
    \\
    \Reify{\SigConst*(T_0,T_1)} &= \lambda t.\ \CPair(\Reify{T_0}(\alpha_{\SigConst*}(t)_0), \Reify{T_1(\alpha_{\SigConst*}(t)_0)}(\alpha_{\SigConst*}(t)_1))
  \end{align*}
  The fact that $\Reify{}$ and $\Reflect{}$ lie over the identity follows directly from the $\beta$
  and $\eta$ laws of dependent sums in \MTT{}. We show the calculations for $\Reflect{}$. Fix $z :
  \Prop$:
  \begin{align*}
    \Reflect{\SigConst*(T_0,T_1)}(e)
    &= \alpha_{\SigConst*}^{-1} \gls{\Reflect{T_0}(\CProj{0}(e)), \Reflect{T_1(\Reflect{T_0}(\CProj{0}(e)))}(\CProj{1}(e))}\\
    &= \alpha_{\SigConst}^{-1} \gls{\CProj{0}(e), \CProj{1}(e)} \\
    &= \alpha_{\SigConst}^{-1} \gls{\alpha_{\SigConst(T_0,T_1)}(e)_0, \alpha_{\SigConst(T_0,T_1)}(e)_1} \\
    &= e
  \end{align*}

  The fact that $\Code{\SigConst(T_0,T_1)}$ and
  $\Pred{\SigConst*(T_0,T_1)}$ lie over $\SigConst(T_0,T_1)$ and $\Tm{m}(z, \SigConst(z, T_0, T_1))$
  respectively follows from definition and realignment.
\end{proof}

\begin{lemma}
  \label{lem:normalization-model:pi}
  Fixing $T_0 : (\mu \mid \Ty*{n})$ and $T_1 : \Fn{\Tm*{n}(T_0)}{\Ty*{m}}$, there exists a pair of constants:
  \begin{align*}
    \PiConst* &: \Ext{\Ty{m}}{z : \Prop}{\PiConst*(z, T_0, T_1)}\\
    \alpha_{\PiConst*} &:
    \Ext{\Tm{m}(\PiConst(T_0, T_1)) \cong \Prod{t : \Modify{\Tm*{n}(T_0)}} \Tm*{m}(T_1(t))}{z : \Prop}{\alpha_{\PiConst}(z,T_0,T_1)}
  \end{align*}
\end{lemma}
\begin{proof}
  Apply \cref{thm:prereq:realignment} to $\alpha_{\PiConst(z)}$ and
  $\FnV{t_0}{\Pred{T_0}}{\Pred{T_1(t_0)}}$ to produce $\Psi : \Uni[1]$ such that
  $z : \Prop \vdash \Psi = \PiConst(z, T_0, T_1)$ and
  $\alpha_{\PiConst*} : \Psi \cong \FnV{t_0}{\Pred{T_0}}{\Pred{T_1(t_0)}}$ which restricts to
  $\alpha_{\PiConst}$ under $z : \Prop$. We now define $\PiConst(T_0,T_1)$ as follows:
  \begin{align*}
    \Code{\PiConst*(T_0,T_1)} &= \CPi(\Code{T_0}, \lambda v.\ \Code{T_1(\Reflect{T_0} v)})
    \\
    \Pred{\PiConst*(T_0,T_1)} &= \Psi
    \\
    \Reflect{\PiConst*(T_0,T_1)} &= \lambda e.\ \alpha_{\PiConst}^{-1}(\lambda t.\ \Reflect{T_1(t)}(\CApp(e)(\Reify{T_0}(t))))
    \\
    \Reify{\PiConst*(T_0,T_1)} &= \lambda t.\ \CLam(\lambda v.\ \Reflect{T_1(\Reflect{T_0}(v))}(\alpha_{\PiConst}(t)(\Reflect{T_0}(v))))
  \end{align*}
  Once again, the fact that everything lies over the correct terms follows from the $\beta$ and
  $\eta$ laws for dependent products, as well as the conclusions of realignment.
\end{proof}

\begin{lemma}
  \label{lem:normalization-model:mod}
  Fixing $T : (\mu \mid \Ty*{n})$, the following constants exist:
  \begin{align*}
    \ModConst* &: \Ext{\Ty{m}}{z : \Prop}{\ModConst(z, T)}
    \\
    \ModIntroConst* &:
    \DeclVar{a}{\Tm*{n}(T)} \to \Ext{\Tm*{m}(\ModConst*(T))}{z : \Prop}{\ModIntroConst(z, T, a)}
    \\
    \ModElimConst* &: \brackets{\DelimMin{1}B : \Modify[\nu]{\Tm*{n}(\ModConst*(T))} \to \Ty*{o}}\\
    &\to \brackets{b : \DelimMin{1}\DeclVar{x}{\Tm*{n}(T)}<\nu\circ\mu> \to B(\ModIntroConst*(T, x))}\\
    &\to \DeclVar{a}{\Tm*{m}(\ModConst*(T))}<\nu>\\
    &\to \Ext{\Tm*{o}(B(a))}{z : \Prop}{\ModElimConst(z, T, B, b, a)}
    \\
    \_ &: \brackets{\DelimMin{1}B : \Modify[\nu]{\Tm*{n}(\ModConst*(T))} \to \Ty*{o}}\\
    &\to \brackets{\DelimMin{1}b : \DeclVar{x}{\Tm*{n}(T)}<\nu\circ\mu> \to \Tm*{o}(B(\ModIntroConst*(T, x)))}\\
    &\to \DeclVar{a}{\Tm*{n}(T)}<\nu\circ\mu>\\
    &\to \ModElimConst*(B, b, \ModIntroConst*(a)) = b(a)
  \end{align*}
\end{lemma}
\begin{proof}
  As is typical now, we will proceed by realigment. We define $\Phi$ to be the realignment of the
  following type along the subsequent isomorphism:
  \begin{align*}
    \Psi
    &= \Sum{m : \Nf{m}(\ModConst(T))}
      \Closed \brackets{
        \parens{\Sum{e : \Ne{m}(\ModConst(T))} \CNfInj(e) = m}
      + \parens{\Sum{a : \Modify{\Pred{T}}} \CMkBox(\Reify{T} a) = m)}
      }
    \\
    \alpha_{\Open} &:
    \Open \Psi \cong \Sum{m : \Open \Nf{m}(\ModConst(T))} \ObjTerm{} \cong \Open[z] \Tm{m}(z, \ModConst(z, T))
  \end{align*}
  This gives a type $\Phi$ and an isomorphism $\alpha : \Phi \cong \Psi$ such that
  $\alpha^{-1}(\gls{m, \dots}) = m$ in a context with $z : \Prop$.

  We now use this to define $\ModConst*(T)$:
  \begin{align*}
    \Code{\ModConst*(T)} &= \CModify(\Code{T})
    \\
    \Pred{\ModConst*(T)} &= \Phi
    \\
    \Reflect{\ModConst*(T)} &= \lambda e. \alpha^{-1}\gls{\CNfInj(e), \Unit(\In{0}(\gls{e, \Ax}))}
    \\
    \Reify{\ModConst*(T)} &= \pi_0 \circ \alpha
  \end{align*}

  Unlike with dependent sums and products, however, this is not the end of the story. We must also
  define the introduction and elimination forms for this type. First, the intro form:
  \[
    \ModIntroConst*(T, m) = \alpha^{-1}\gls{\CMkBox(\Reify{T}(m)), \Unit(\In{1}(m, \Ax))}
  \]
  In order to define the elimination form, we need the elimination principle for $\Closed A$. Recall
  $\Closed A = \Prop \Join A$ so we may use the induction principle of a pushout on $\Closed A$.

  With this in mind, we define the elimination rule $\ModElimConst*(T, T_m, b, s)$ as
  follows. First, recall that $s : \Modify[\nu]{\Tm*{m}(\ModConst*(T))}$. We use the fact that
  $\Modify[\nu]{-}$ preserves dependent sums, the closed modality (\cref{thm:prereq:model}), and
  coproducts (\cref{thm:prereq:internal-left-adj}) to decompose $s$ into two terms:
  \begin{align*}
    m &: \Modify[\nu]{\Nf{n}(\ModConst(T))}
    \\
    q &:
    \Closed
    \brackets{\Sum{e : \Modify[\nu]{\Ne{m}(\ModConst(T))}} \Modify[\nu]{\CNfInj(e) = m}
    + \Sum{a : \Modify[\nu\circ\mu]{\Pred{T}}} \Modify[\nu]{\CMkBox(\Reify{T} a) = m})}
  \end{align*}
  We may now perform induction on $q$:
  \[
    \begin{cases}
      \ModElimConst(z, T, T_m, b, s) & q = \In{0}(z)
      \\
      \Reify{T_m(s)}\ \CLetMod(T, T_m, \lambda x.\ \Reify{T_m(\ModIntroConst(\Reflect{T} x))} b(\Reflect{T} x), e)
      & q = \In{1}(\In{0}(\MkBox[\nu]{e}, \_))
      \\
      b(a) & q = \In{1}(\In{1}(\MkBox[\nu]{a}, \_))
    \end{cases}
  \]
  The fact that these agree on overlaps follows from the $\beta$ rule for modal types from
  \MTT{}. By construction, the elimination and introduction forms lie over their syntactic
  counterparts, and calculation shows that the $\beta$ equation holds.
\end{proof}

\begin{lemma}
  \label{lem:normalization-model:bool}
  The following constants exist:
  \begin{align*}
    \BoolConst* &: \Ext{\Ty*{m}}{z : \Prop}{\BoolConst(z)}
    \\
    \TrueConst* &: \Ext{\Tm*{m}(\BoolConst)}{z : \Prop}{\TrueConst}
    \\
    \FalseConst* &: \Ext{\Tm*{m}(\BoolConst)}{z : \Prop}{\FalseConst}
    \\
    \IfConst* &: (T : \Tm*{m}(\BoolConst(z)) \to \Ty*{m}) \to\\
    &\qquad \Tm*{m}(T(\TrueConst*)) \to \Tm*{m}(T(\FalseConst*)) \to (b : \Tm*{m}(\BoolConst*)) \to\\
    &\qquad \braces{\Tm*{m}(T(b)) \mid z : \Prop \mapsto \IfConst(T, t, f, b)}
    \\
    \_ &: (T : \Tm*{m}(\BoolConst*) \to \Ty*{m})(t : \Tm*{m}(T(\TrueConst*)))(f : \Tm*{m}(T(\FalseConst*))) \to\\
    &\qquad (\IfConst*(T, t, f, \TrueConst*) = t) \times (\IfConst*(T, t, f, \FalseConst*) = f)
  \end{align*}
\end{lemma}
\begin{proof}
  The proof proceeds much as \cref{lem:normalization-model:mod}, so we present the relevant definitions
  without much commentary. First, we realign $\Psi$ along $\alpha_{\Open}$:
  \begin{align*}
    \Psi &= \Sum{m : \Nf{m}(\BoolConst*)}
    \Closed \brackets{
      \parens{\Sum{e : \Ne{m}(\BoolConst)} \CNfInj(e) = m}
      +
      \parens{\Sum{b : \mathbf{2}} m = \mathsf{rec}_{\mathbf{2}}(b; \NfTrue; \NfFalse))}
    }
    \\
    \alpha_{\Open} &: \Open[z] \Psi \cong \Sum{m : \Nf{m}(\BoolConst*)} \ObjTerm{} \cong \Tm{m}(z, \BoolConst(z))
  \end{align*}
  As a result, we obtain $\Phi$ along with $\alpha : \Phi \cong \Psi$ such that
  $z : \Prop \vdash \alpha^{-1}\gls{m, \dots} = m$. We may now define $\BoolConst*$:
  \begin{align*}
    \Code{\BoolConst*} &= \CBool
    \\
    \Pred{\BoolConst*} &= \Phi
    \\
    \Reflect{\BoolConst*} &= \lambda e. \alpha^{-1}\gls{\CNfInj(e), \Unit(\In{0}(e, \Ax))}
    \\
    \Reify{\BoolConst*} &= \pi_0 \circ \alpha
  \end{align*}

  The true and false constants are defined as follows:
  \begin{align*}
    \TrueConst* &= \gls{\CTrue, \Unit(\In{1}(0, \Ax))}\\
    \FalseConst* &= \gls{\CFalse, \Unit(\In{1}(1, \Ax))}
  \end{align*}

  Again, the elimination principle is defined using the induction principle for $\Closed A$.
  \begin{align*}
    \IfConst*&(T_m, t_0, t_1, s = \alpha^{-1} \gls{b, q}) =\\
    &
    \begin{cases}
      \IfConst(z, T_m, t_0, t_1, m) & q = \In{0}(z)
      \\
      \Reify{T_m(s)} \CIf(T_m, \Reify{T_m(\TrueConst*)} t_0, \Reify{T_m(\FalseConst*)} t_1, e)
      & q = \In{1}(\In{0}(e, \Ax)), \CNfInj(e) = b
      \\
      t_i & q = \In{1}(\In{1}(i, \Ax)), \mathbf{rec}(i; \NfTrue; \NfFalse) = b
    \end{cases}
  \end{align*}
  The boundary conditions and the computation rule follow from computation.
\end{proof}

\begin{lemma}
  \label{lem:normalization-model:uni}
  The glued model supports a universe. Specifically, this means that the following constants exist:
  \begin{align*}
    \UniConst* &: \Ext{\Ty*{m}}{z : \Prop}{\UniConst}
    \\
    \DecConst* &: (T : \Tm*{m}(\UniConst*)) \to \Ext{\Ty*{m}}{z : \Prop}{\DecConst(T)}
  \end{align*}

  Moreover, $\UniConst*$ is closed under various type-formers:
  \begin{align*}
    \SigCodeConst* &:
    \Ext{
      \brackets{\Sum{A : \Tm*{m}(\UniConst*)} \Tm*{m}(\DecConst*(A)) \to \Tm*{m}(\UniConst)} \to \Tm*{m}(\UniConst*)
    }{z : \Prop}{\SigCodeConst}
    \\
    \PiCodeConst* &:
    \Ext{
      \brackets{\Sum{A : \Modify{\Tm*{n}(\UniConst*)}} \brackets{\Modify{\Tm*{n}(\DecConst*(A))} \to \Ty*{m}}} \to \Tm*{m}(\UniConst*)
    }{z : \Prop}{\PiCodeConst}
    \\
    \BoolCodeConst* &: \Ext{\Tm*{m}(\UniConst*)}{z : \Prop}{\BoolCodeConst}
    \\
    \ModCodeConst* &: \Ext{\Modify{\Tm*{n}(\UniConst*)} \to \Tm*{m}(\UniConst*)}{z : \Prop}{\ModCodeConst}
    \\
    \DecIsoConst*_{\SigCodeConst} &:
    (A : \Tm*{m}(\UniConst*))(B : \Tm*{m}(\DecConst*(A)) \to \Tm*{m}(\UniConst*)) \to\\
    &\qquad \Tm*{m}(\DecConst*(\SigCodeConst*(A,B))) \cong \Tm*{m}(\SigConst*(\DecConst*(A), \DecConst* \circ B))
    \\
    \DecIsoConst*_{\PiCodeConst} &:
    \DeclVar{A}{\Tm*{m}(\UniConst*)}(B : \Fn{\Tm*{n}(\DecConst*(A))}{\Tm*{m}(\UniConst*)}) \to\\
    &\qquad \Tm*{m}(\DecConst*(\PiCodeConst*(A,B))) \cong \Tm*{m}(\PiConst*(\DecConst*(A), \DecConst* \circ B))
    \\
    \DecIsoConst*_{\BoolCodeConst} &: \Tm*{m}(\DecConst*(\BoolCodeConst*)) \cong \Tm*{m}(\BoolConst*)
    \\
    \DecIsoConst*_{\ModCodeConst} &:
    \DeclVar{A}{\Tm*{m}(\UniConst*)} \to
    \Tm*{m}(\DecConst*(\ModCodeConst*(A))) \cong \Tm*{m}(\ModConst*(\DecConst*(A)))
  \end{align*}
  Furthermore, when given $z : \Prop$ we require the following equations:
  \begin{align*}
    \DecIsoConst*_{\SigCodeConst} &= \DecIsoConst_{\SigCodeConst}(z)
    &
    \DecIsoConst*_{\PiCodeConst} &= \DecIsoConst_{\PiCodeConst}(z)
    &
    \DecIsoConst*_{\BoolCodeConst} &= \DecIsoConst_{\BoolCodeConst}(z)
    &
    \DecIsoConst*_{\ModCodeConst} &= \DecIsoConst_{\ModCodeConst}(z)
  \end{align*}
\end{lemma}
\begin{proof}
  At this point we take advantage of the fact that $\Pred{\Ty*{m}}$ is an element of $\Uni[1]$; in
  particular, we use the fact that is a universe $\Uni[0]$ small enough to fit inside $\Uni[1]$.

  \paragraph{The universe and $\Dec{-}$ constant.}
  We may then define $\Psi$ by realigning the following element of $\Uni[1]$ along the evident
  isomorphism to $\Tm*{m}(z,\UniConst(z))$:
  \begin{equation}
    \begin{make-rcd}{T}{\Uni[1]}
      \Code{} : \Nf{m}(\UniConst)
      \\
      \Pred{} : \Ext{\Uni[0]}{z : \Prop}{\Tm{m}(z, \DecConst(z, \Code{}))}
      \\
      \Reflect{} : \Ext{\Ne{m}(\Open[z] \DecConst(z, \Code{})) \to \Pred{}}{z : \Prop}{\ArrId{}}
      \\
      \Reify{} : \Ext{\Pred{} \to \Nf{m}(\Open[z] \DecConst(z, \Code{}))}{z : \Prop}{\ArrId{}}
    \end{make-rcd}
  \end{equation}
  As a result, we obtain $\Psi : \Uni[1]$ and $\alpha : \Psi \cong T$ such that
  $\alpha^{-1}(\braces{\Code{}; \dots}) = \Code{}$ given $z : \Prop$. With $\Psi$ in hand, we may
  define $\UniConst*$:
  \begin{align*}
    \Code{\UniConst*} &= \CUni
    \\
    \Pred{\UniConst*} &= \Psi
    \\
    \Reflect{\UniConst*} &= \lambda e.\ \alpha^{-1}\gls{\CNfInj(e); \Ne{m}; \ArrId{}; \lambda e.\ \CNfInj(e)}
    \\
    \Reify{\UniConst*} &= \pi_0 \circ \alpha
  \end{align*}

  The definition of $\DecConst*$ is essentially ``just'' cumulativity:
  \[
    \DecConst*(\alpha^{-1}\gls{\Code{}; \Pred{}; \Reify{}; \Reflect{}})
    =
    \alpha\gls{\NfDec{\Code{}}; \Pred{}; \Reify{}; \Reflect{}}
  \]
  Here we have used both the isomorphism realigning $\UniConst*$ and the isomorphism realigning
  $\Ty*{m}$. When restricting with $z : \Prop$, $\DecConst*$ restricts to $A \mapsto \DecConst(A)$.

  It remains to show that the universe is closed under the various codes for dependent sums,
  products, \etc{} For the sake of space, we show only two of these cases: dependent products and
  modal types.

  Let us fix $t_0 : (\mu \mid \Tm*{n}(\UniConst*))$ and
  $t_1 : \Fn{\Tm*{n}(\DecConst*(t_0))}{\Tm*{m}(\UniConst*)}$. First we realign
  $\Tm*{m}(\PiConst*(\DecConst*(A), \DecConst* \circ B))$ along $\DecIsoConst_{\PiCodeConst}$ to
  obtain $\Psi$ which lies strictly over
  $\Tm*{m}(z, \PiConst(z, \DecConst(z, t_0), \DecConst(z) \circ t_1))$.

  We also get an isomorphism
  $\DecIsoConst*_{\PiCodeConst} : \Psi \cong \Tm*{m}(\PiConst*(\DecConst*(A), \DecConst* \circ B))$
  lying over $\DecIsoConst$. Note that by the uniqueness of inverses, $(\DecIsoConst*)^{-1}$ lies over
  $\DecIsoConst^{-1}$.

  We use this now to define $\PiCodeConst*(t_0,t_1)$:
  \begin{align*}
    \Code{\PiCodeConst*(t_0,t_1)} &= \FnCode{\Code{t_0}}{\lambda x.\ \Code{t_1(\Reflect{t_0}(x))}}
    \\
    \Pred{\PiCodeConst*(t_0,t_1)} &= \Psi
    \\
    \Reflect{\PiCodeConst*(t_0,t_1)} &= \lambda e.\ (\DecIsoConst*_{\PiCodeConst})^{-1}(\Reflect{\PiConst*(\DecConst*(t_0), \DecConst* \circ t_1)} \NfDecIso{e})
    \\
    \Reify{\PiCodeConst*(t_0,t_1)} &= \lambda f.\ \NfDecIso*{\Reify{\PiConst*(\DecConst*(t_0), \DecConst* \circ t_1)}\DecIsoConst*_{\PiCodeConst}(f)}
  \end{align*}

  The procedure for modal types is similar. Let us fix $t : (\mu \mid \Tm*{n}(\UniConst*))$. Again,
  we realign $\Tm*{m}(\ModConst*(\DecConst*(t)))$ along the isomorphism $\DecIsoConst_{\ModCodeConst}$
  to obtain $\Psi$ and $\DecIsoConst*_{\ModConst}$. The actual construction of $\ModCodeConst*$ is
  almost identical to $\PiCodeConst*$:
  \begin{align*}
    \Code{\ModCodeConst*(t)} &= \ModifyCode{\Code{t}}
    \\
    \Pred{\ModCodeConst*(t)} &= \Psi
    \\
    \Reflect{\ModCodeConst*(t)} &= \lambda e.\ (\DecIsoConst*_{\ModCodeConst})^{-1}(\Reflect{\ModConst*(\DecConst*(t))} \NfDecIso{e})
    \\
    \Reify{\ModCodeConst*(t)} &= \lambda m.\ \NfDecIso*{\Reify{\PiConst*(\DecConst*(t))} \DecIsoConst*_{\ModCodeConst}(m)}
  \end{align*}
  The checks that all constructions lie over their syntactic counterparts follow immediately from
  the conclusions of realignment.
\end{proof}

\begin{theorem}
  \label{thm:normalization-model:model}
  There exists an \MTT{} cosmos built around $\Ty*{m}$ and $\Tm*{m}$ which lies strictly over the
  syntactic \MTT{} cosmos.  Externally, there is an \MTT{} cosmos
  $\Mor[\InterpGl{-}]{\Mode}{\GL{\InvRen[-]}}$ together with a morphism of cosmoi
  $\Mor[\pi]{\InterpGl{-}}{\InterpSyn{-}}$.
\end{theorem}
\begin{proof}
  This is combination of
  \cref{%
    def:normalization-model:ty,%
    def:normalization-model:tm,%
    lem:normalization-model:tm-align,%
    lem:normalization-model:sig,%
    lem:normalization-model:pi,%
    lem:normalization-model:mod,%
    lem:normalization-model:bool,%
    lem:normalization-model:uni%
  }.
  That $\pi$ is 2-natural is immediate from \cref{thm:foundations:induced} and that it commutes with
  transposition is precisely \cref{thm:foundations:essential}.
\end{proof}


\section{The normalization function}
\label{sec:normalization}

\subsection{Initiality, revisited}

As formulated, the cosmos based around the generalized algebraic syntax of \MTT{} $\InterpSyn{-}$ is
not initial. By \cref{cor:cosmoi:initial}, there is an initial cosmos $\ICosmos{-}$ but it is a
distinct object and $\ICosmos{m}$ is not even cocomplete, let alone a presheaf topos.

Despite this, $\InterpSyn{-}$ enjoys a certain distinguished place among \MTT{} cosmoi owing to
\cref{thm:cosmoi:quasi-initiality}. In particular, we have the following
\begin{enumerate}
\item For every context $\IsCx{\Gamma}$, there exists an object $\InterpGl{\Gamma} : \InterpGl{m}$,
  together with a canonical isomorphism $\alpha_\Gamma : \pi(\InterpGl{\Gamma}) \cong \Yo{\Gamma}$.
\item For every type $\IsTy{A}$, there is a morphism
  $\Mor[\InterpGl{A}]{\InterpGl{\Gamma}}{\Ty*{m}}$ such that
  $\pi(\InterpGl{A}) \circ \alpha_\Gamma = \YoEm{A}$.
\item For every term $\IsTm{M}{A}$, there is a morphism $\Mor[\InterpGl{M}]{\InterpGl{M}}{\Tm*{m}}$
  lying over $\InterpGl{A}$ such that $\pi(\InterpGl{M}) \circ \alpha_\Gamma = \YoEm{M}$.
\end{enumerate}

\subsection{The proof of the normalization theorem}

Prior to the proof of the normalization theorem (\cref{thm:normalization:normalization}), we require
several preliminary results.

Recall that there is a closed immersion of topoi $\Mor[\ClIncl]{\PSH{\Ren{m}}}{\InterpGl{m}}$. In
fact, this closed immersion is essential because the morphism $\Mor{\PSH{\Ren{m}}}{\PSH{\Cx{m}}}$ is
essential. Accordingly, for each $X : \PSH{\Ren{m}}$ there is an object
$\ClIncl_!(X) : \InterpGl{m}$. This functor has a familiar result when applied to a representable in
$\PSH{\Ren{m}}$. Consider a telescope $\IsTele{\Theta}{\Gamma}$, we then have the following:
\[
  \AtomicCx{\Theta} \defeq \Ess{\ClIncl}(\Yo{\Theta}) = \Mor{\Yo{\Theta}}{\InvRen{\Yo{\Gamma}}}
\]
\begin{lemma}
  \label{lem:normalization:atomic-yoneda}
  For an arbitrary $X : \InterpGl{m}$, we have
  $\Hom{\AtomicCx{\Theta}}{X} \cong \Inv{\ClIncl}(X)(\Theta)$.
\end{lemma}

By the previous discussion, for a telescope $\IsTele{\Theta}{\Gamma}$, there is also the object
$\InterpGl{\Gamma}$, which lies over $\Gamma$ in $\InterpGl{m}$. While this object is different
than $\AtomicCx{\Theta}$, there is a canonical natural transformation between them.

\begin{lemma}
  \label{lem:normalization:world}
  Given $\IsTele{\Theta}{\Gamma}$, there is a morphism
  $\Mor[\Atoms{\Theta}]{\AtomicCx{\Theta}}{\InterpGl{\Gamma}}$.
  Moreover, up to the canonical isomorphism $\Inv{\OpenIncl}{\InterpGl{\Gamma}} \cong \Yo{\Gamma}$,
  $\Inv{\OpenIncl}{\Atoms{\Theta}} = \ArrId{\Yo{\Gamma}}$.
\end{lemma}
\begin{proof}
  By applying \cref{lem:normalization:atomic-yoneda}, this proof is equivalent to constructing an
  element of $g \in \Inv{\ClIncl}{\InterpGl{\Gamma}}(\Theta)$ such that the image of $g$ under the
  canonical map to $\InvRen{\Yo{\Gamma}}(\Theta)$ is the identity. This morphism is constructed by
  induction on $\Theta$.
  \begin{description}
  \item[Case]
    \[
      \inferrule{ }{
        \IsTele{\EmpTele}{\EmpCx}
      }
    \]
    In this case these two are isomorphic.

  \item[Case]
    \[
      \inferrule{
        \Mor[\mu]{n}{m}
        \\
        \IsTele{\Theta}{\Gamma}
      }{
        \IsTele{\LockTele{\Theta}}<n>{\LockCx{\Gamma}}
      }
    \]
    In this case we choose $\Atoms{\LockTele{\Theta}} = \Ess{\InterpGl{\mu}}(\Atoms{\Theta_0})$ and
    correcting by the canonical isomorphisms
    $\AtomicCx{\LockTele{\Theta_0}} \cong \Ess{\InterpGl{\mu}}(\AtomicCx{\Theta})$ which exists by
    \cref{thm:foundations:essential} and
    $\InterpGl{\LockCx{\Gamma}} \cong \Ess{\InterpGl{\mu}}(\InterpGl{\Gamma})$ which exists by
    construction.
  \item[Case]
    \[
      \inferrule{
        \IsTele{\Theta}{\Gamma}
        \\
        \IsTy[\LockCx{\Gamma}]{A}<n>
      }{
        \IsTele{\ETele{\Theta}{A}}{\ECx{\Gamma}{A}}
      }
    \]
    We construct this morphism using the characterization from
    \cref{lem:normalization:atomic-yoneda}. Our induction hypothesis gives us an element
    $g \in \InterpGl{\Gamma}_0(\Theta)$ lying over $\alpha_\Gamma^{-1}(\ISb)$, whence
    $g' \in \InterpGl{\ECx{\Gamma}{A}}(\ETele{\Theta}{A})$ lying over $\Wk$ using reindexing.

    Moreover, we have an element of $\Ne{m}(\InterpSyn{A})(\LockTele{\ETele{\Theta}{A}})$ lying over
    $\Var{0}$, whence an element of $\Tm*{m}(\InterpGl{A})(\LockTele{\ETele{\Theta}{A}})$ lying over
    $\Var{0}$. This gives us an element of $\InterpGl{\ECx{\Gamma}{A}}_0(\ETele{\Theta}{A})$ which
    lies over $\ISb = \ESb{\Wk}{\Var{0}}$ by the universal property of the pullback.
    \qedhere
  \end{description}
\end{proof}

\begin{lemma}
  \label{lem:normalization:same-nf}
  Given $\IsTm{M,N}{A}$, if there exists a normal form $u$ such that $\IsNf{u}{A}{M}$ and
  $\IsNf{u}{A}{N}$, then $\EqTm{M}{N}{A}$.
\end{lemma}
\begin{lemma}
  \label{lem:normalization:same-nfty}
  If $\IsNfTy{u}{A}$ and $\IsNfTy{u}{B}$, then $\EqTy{A}{B}$.
\end{lemma}

We are now in a position to prove normalization. We note that one of the conditions of the
normalization algorithm (respect for definitional equivalence) is automatic: we have defined the
algorithm to operate only on equivalence classes of terms, so it must always respect definitional
equivalence.

\begin{theorem}
  \label{thm:normalization:normalization}
  There exists a pair of functions $\Normalize{M}{A}$ and $\NormalizeTy{A}$ such that
  \begin{enumerate}
  \item If $\IsTy{A}$ and $\IsTele{\Theta}{\Gamma}$ then $\IsNfTy{\NormalizeTy{A}}{A}$.
  \item If $\IsTm{M}{A}$ and $\IsTele{\Theta}{\Gamma}$ then $\IsNf{\Normalize{M}{A}}{A}{M}$.
  \end{enumerate}
\end{theorem}
\begin{proof}
  We may package this in the language of $\PSH{\Ren{m}}$. From this viewpoint,
  $\NormalizeTy[\Gamma]{A}$ is an element of $(\NfTy{m})_0$ fitting into the following triangle:
  \begin{equation*}
    \begin{tikzpicture}[diagram]
      \node (A) {$\InvRen{\Yo{\Gamma}}$};
      \node[right = 4cm of A] (B) {$\InvRen{\Yo{\Gamma}}$};
      \node[above = 1.5cm of B] (C) {$(\NfTy{m})_0$};
      \path[->] (A) edge node[below] {$\YoEm{A}$} (B);
      \path[->] (A) edge (C);
      \path[->] (C) edge node[right] {$\NfTy{m}$} (B);
    \end{tikzpicture}
  \end{equation*}
  Similarly, $\Normalize{M}{A}$ is a section to the inclusion $\Mor{\Nf{m}}{\InvRen{\SynTm{m}}}$
  over $\YoEm{M}$.

  Observe from the initiality of syntax that for any $\IsTy{A}$, there exists
  $\Mor[\InterpGl{A}]{\InterpGl{\Gamma}}{\Ty*{m}}$, and bi-universality ensures that
  $\Open \InterpGl{A} = A$, up to an isomorphism of contexts. In particular, there exists
  $\Mor[A^{\mathbf{nf}}]{\InterpGl{\Gamma}}{(\NfTy{m})_0}$ such that $\Open A^{\mathbf{nf}} = A$.
  Unfolding to $\PSH{\Ren{m}}$, we have the following commuting square:
  \[
    \DiagramSquare{
      width = 4.5cm,
      height = 1.5cm,
      nw = \InterpGl{\Gamma}_0,
      sw = \InvRen{\Yo{\Gamma}},
      west = \InterpGl{\Gamma},
      ne = (\NfTy{m})_0,
      se = \InvRen{\SynTy{m}},
      east = \NfTy{m},
      south = \InvRen{\YoEm{A}},
      north = A^{\mathbf{nf}},
    }
  \]
  By \cref{lem:normalization:world}, $\Atoms{\Theta} \in \InterpGl{\Gamma}_1(\Theta)$ lies over
  $\ISb$. We then define $\NormalizeTy{A} = A^{\mathbf{nf}}(\Atoms{\Theta})$, which satisfies the
  required properties.

  The procedure is identical for $\Normalize{M}{A}$. Given $\IsTm{M}{A}$, we use initiality to
  obtain the following square:
  \[
    \DiagramSquare{
      width = 4.5cm,
      height = 1.5cm,
      nw = \InterpGl{\Gamma}_0,
      sw = \InvRen{\Yo{\Gamma}},
      west = \InterpGl{\Gamma},
      ne = (\Nf{m})_0,
      se = \InvRen{\SynTm{m}},
      east = \Nf{m},
      south = \InvRen{\YoEm{M}},
      north = M^{\mathbf{nf}},
    }
  \]
  We set $\Normalize{M}{A} = M^{\mathbf{nf}}(\Atoms{\Theta})$.
\end{proof}

\begin{theorem}
  Given a telescope $\IsTele{\Theta}{\Gamma}$ then the following two facts hold:
  \begin{enumerate}
  \item If $\IsNf{u}{A}{M}$, then $\Normalize{M}{A} = u$.
  \item If $\IsNfTy{\tau}{A}$, then $\NormalizeTy{A} = \tau$.
  \end{enumerate}
\end{theorem}
\begin{proof}
  In order to prove these results, we show three related facts. Recall that $\InterpGl{-}$ is the
  function sending a piece of syntax to its interpretation in the normalization model. Furthermore,
  recall that by \cref{lem:normalization:world} $\AtomicCx{\Theta}$ element
  $\Atoms{\Theta} : \InterpGl{\Gamma}$.
  \begin{enumerate}
  \item If $\IsNe{e}{A}{M}$, then
    $\InterpGl{M}(\Atoms{\Theta}) = \Reflect{\InterpGl{A}(\Atoms{\Theta})} e$
  \item If $\IsNf{u}{A}{M}$, then
    $\Reflect{\InterpGl{A}(\Atoms{\Theta})}\InterpGl{M}(\Atoms{\Theta}) = u$.
  \item If $\IsNfTy{\tau}{A}$, then $\Code{\InterpGl{A}(\Atoms{\Theta})} = \tau$.
  \end{enumerate}
  Here we have identified a code $u$ (resp. $e$) as an $\AtomicCx{\Theta}$ element of $\Nf{A}$
  (resp. $\Ne{A}$), an abuse justified by unfolding the definition of $\Nf{A}$ into $\PSH{\Ren{m}}$.

  Assuming these facts hold, the result immediately follows by unfolding the definition of
  $\Normalize{M}{A}$ and $\NormalizeTy{A}$. We prove these facts by mutual induction. All cases
  follow from induction except the case of variables, so we show this here.

  \begin{description}
  \item[Case]
    \[
      \inferrule{
        \IsTele{\Theta}{\Gamma}
        \\
        \Theta = \ETele{\Theta_0}{A}
      }{
        \IsNe{\NeVar{0}{\ArrId{\mu}}}{
          \Sb{A}{\Key{\alpha}{} \circ \LockSb{\Wk}}
        }{
          \Sb{\Var{0}}{\Key{\alpha}{}}
        }
      }
    \]
    (Note that the case for a general $\NeVar{k}{\alpha}$ follows by straightforward induction on
    $k$, but is notationally heavy.)

    First, we recall that $\InterpGl{\Var{0}}$ defined by projection from the context. Therefore,
    $\InterpGl{\Var{0}}(\Atoms{\Theta}) = \Reflect{\InterpGl{A}(\Atoms{\Theta})} \NeVar{0}{\ArrId{\mu}}$
    by unfolding the construction of $\Atoms{\Theta}$ in
    \cref{lem:normalization:world}.

  \item[Case]
    \[
      \inferrule{
        \IsNfTy[\LockTele{\Theta}]{\tau}<n>{A}
      }{
        \IsNfTy{\NfModify{\tau}}{\Modify{A}}
      }
    \]
    In this case, by unfolding the definitions we have
    \begin{align*}
      \Code{\InterpGl{\Modify{A}}(\Atoms{\LockTele{\Theta}})}
      &= \CModify(\Code{A(\Atoms{\Theta})})
    \end{align*}
    By induction hypothesis, we have $\Code{A(\Atoms{\Theta})} = \tau$, so the conclusion follows.

  \item[Case]
    \[
      \inferrule{
        \IsNf[\LockTele{\Theta}]{u}{A}<n>{M}
      }{
        \IsNf{\NfMkBox{u}}{\Modify{A}}{\MkBox{M}}
      }
    \]
    In this case, by unfolding the definitions we have
    \begin{align*}
      \Reify{\InterpGl{\Modify{A}}(\Atoms{\LockTele{\Theta}})} & \InterpGl{\MkBox{M}}(\Atoms{\LockTele{Theta}})\\
      &= \CMkBox{\Reify{\InterpGl{A}(\Atoms{\Theta})} \InterpGl{M}(\Atoms{\Theta})}\\
      &= \CMkBox(u)
    \end{align*}
    The last step follows from the induction hypothesis.
    \qedhere
  \end{description}
\end{proof}

\begin{corollary}
  \label{cor:normalization:unique}
  Each term and type in \MTT{} has a unique normal form.
\end{corollary}

\begin{corollary}
  \label{cor:normalization:conversion}
  The conversion problem in \MTT{} is equivalent the conversion problem of normal forms.
\end{corollary}
\begin{proof}
  \Cref{thm:normalization:normalization,lem:normalization:same-nf} imply that $\EqTm{M}{N}{A}$ is
  equivalent to $\Normalize{M}{A} = \Normalize{N}{A}$. Similarly,
  \Cref{thm:normalization:normalization,lem:normalization:same-nfty} reduce $\EqTy{A}{B}$ to
  $\NormalizeTy{A} = \NormalizeTy{B}$. As the proof given is constructively valid, these reductions
  are effective.
\end{proof}

\begin{corollary}
  \label{cor:normalization:pi-inj}
  If $\EqTy{\Fn{A_0}{B_0}}{\Fn{A_1}{B_1}}$, then $\EqTy[\LockCx{\Gamma}]{A_0}{A_1}$ and
  $\EqTy[\ECx{\Gamma}{A_0}]{B_0}{B_1}$.
\end{corollary}
\begin{proof}
  By inspection on the definition of $\NormalizeTy{\Fn{A_i}{B_i}}$, we observe that both are
  interpreted using \cref{lem:normalization-model:pi}. Accordingly, we have the following:
  \begin{align*}
    \NormalizeTy{\Fn{A_0}{B_0}} &= u = \NfFn{\NormalizeTy[\LockCx{\Gamma}]{A_0}}{\NormalizeTy[\ECx{\Gamma}{A_0}]{B_0}}\\
    \NormalizeTy{\Fn{A_1}{B_1}} &= v = \NfFn{\NormalizeTy[\LockCx{\Gamma}]{A_1}}{\NormalizeTy[\ECx{\Gamma}{A_1}]{B_1}}
  \end{align*}
  From $\EqTy{\Fn{A_0}{B_0}}{\Fn{A_1}{B_1}}$ and \cref{cor:normalization:conversion}, we obtain
  $u = v$. From this, we use inversion to conclude that
  $\NormalizeTy[\LockCx{\Gamma}]{A_0} = \NormalizeTy[\LockCx{\Gamma}]{A_1}$ and
  $\NormalizeTy[\ECx{\Gamma}{A_0}]{B_0} = \NormalizeTy[\ECx{\Gamma}{A_1}]{B_1}$. The result now
  follows from \cref{cor:normalization:conversion}.
\end{proof}

\begin{corollary}
  If modalities and 2-cells enjoy decidable equality, typechecking \MTT{} is decidable.
\end{corollary}


\section*{Acknowledgments}

I am thankful for discussions with Carlo Angiuli, Martin Bidlingmaier, Lars Birkedal, Thierry
Coquand, Alex Kavvos, Christian Sattler, and Jonathan Sterling. I would like to express particular
gratitude to Jonathan Sterling conversations about Synthetic Tait Computability and to Lars Birkedal
and Alex Kavvos for feedback on this proof.

\appendix
\section{Neutral and normal forms}
\label{sec:appendix:neutrals-and-normals}

\begin{mathparpagebreakable}
  \JdgFrame{\IsTele{\Theta}{\Gamma}}
  \\
  \inferrule{ }{
    \IsTele{\EmpTele}{\EmpCx}
  }
  \and
  \inferrule{
    \IsTele{\Theta}{\Gamma}
    \\
    \IsTy[\LockCx{\Gamma}]{A}<n>
  }{
    \IsTele{\ETele{\Theta}{A}}{\ECx{\Gamma}{A}}
  }
  \and
  \inferrule{
    \IsTele{\Theta}{\Gamma}
  }{
    \IsTele{\LockTele{\Theta}}{\LockCx{\Gamma}}<n>
  }
  \\
  \JdgFrame{\IsRen{\psi}{\Psi}{\gamma}}
  \\
  \inferrule{ }{
    \IsRen{\EmpRen}{\EmpTele}{\EmpSb}
  }
  \and
  \inferrule{
    \IsTele{\Theta}<n>{\Gamma}
    \\
    \IsTy[\LockCx{\Gamma}]{A}<n>
  }{
    \IsRen[\ETele{\Theta}{A}]{\WkRen}{\Theta}{\Wk}
  }
  \and
  \inferrule{
    \IsTele{\Theta}{\Gamma}
  }{
    \IsRen{\IRen}{\Theta}{\ISb}
  }
  \and
  \inferrule{
    \IsTele{\Gamma, \Delta,\Xi}{\Theta,\Psi,\Phi}
    \\
    \IsRen[\Theta]{r}{\Psi}{\gamma}
    \\
    \IsRen[\Psi]{s}{\Phi}{\delta}
  }{
    \IsRen[\Theta]{s \circ r}{\Phi}{\delta\circ\gamma}
  }
  \and
  \inferrule{
    \IsTele{\Gamma, \Delta}{\Theta,\Psi}
    \\
    \IsRen{r}{\Psi}{\delta}
  }{
    \IsRen[\LockTele{\Theta}]{\LockRen{r}}{\LockTele{\Psi}}<n>{\LockSb{\delta}}
  }
  \and
  \inferrule{
    \IsTele{\Theta}{\Gamma}
    \\
    \Mor[\mu, \nu]{n}{m}
    \\
    \Mor[\alpha]{\nu}{\mu}
  }{
    \IsRen[\LockTele{\Theta}]{\KeyRen{\alpha}{\Theta}}{\LockTele{\Theta}<\nu>}<n>{\Key{\alpha}{\Gamma}}
  }
  \and
  \inferrule{
    \IsTele{\Gamma, \Delta}{\Theta,\Psi}
    \\
    \IsRen{r}{\Psi}<m>{\delta}
    \\
    \IsTy[\LockCx{\Delta}]{A}<n>
    \\
    \IsNe[\LockCx{\Gamma}]{\NeVar{k}{\alpha}}{\Sb{A}{\LockSb{\delta}}}<n>{M}
  }{
    \IsRen{\ERen{r}{\NeVar{k}{\alpha}}}{\ETele{\Psi}{A}}{\ESb{\delta}{M}}
  }
  \\
  \JdgFrame{\IsNfTy{\tau}{A}}
  \\
  \inferrule{ }{
    \IsNfTy{\NfBool}{\Bool}
    \and
    \IsNfTy{\NfUni}{\Uni}
  }
  \and
  \inferrule{
    \IsNfTy[\LockTele{\Theta}]{\tau}<n>{A}
    \\
    \IsNfTy[\ETele{\Theta}{A}]{\sigma}{B}
  }{
    \IsNfTy{\NfFn{\tau}{\sigma}}{\Fn{A}{B}}
  }
  \and
  \inferrule{
    \IsNfTy{\tau}<n>{A}
    \\
    \IsNfTy[\ETele{\Theta}{A}<\ArrId{m}>]{\sigma}{B}
  }{
    \IsNfTy{\NfProd{\tau}{\sigma}}{\Sig{A}{B}}
  }
  \and
  \inferrule{
    \IsNfTy[\LockTele{\Theta}]{\tau}<n>{A}
  }{
    \IsNfTy{\NfModify{\tau}}{\Modify{A}}
  }
  \and
  \inferrule{
    \IsTele{\Theta}{\Gamma}
    \and
    \IsNf{u}{\NfUni}{A}
  }{
    \IsNfTy{\NfDec{u}}{A}
  }
  \\
  \JdgFrame{
    \IsNe{e}{A}{M}
    \qquad
    \IsNf{u}{A}{M}
  }
  \\
  \inferrule{
    \IsTele{\Theta}{\Gamma}
    \\
    \Theta(k) = (\DeclNameless{A})
    \\
    \Locks{\Theta}{k} = \nu
    \\
    \Mor[\alpha]{\mu}{\nu}
  }{
    \IsNe{\NeVar{k}{\alpha}}{
      \Sb{A}{\Key{\alpha}{} \circ (\LockSb{\Wk}<\nu_{k-1}>) \dots \circ (\LockSb{\Wk}<\nu_{0}>)}
    }{
      \Sb{\Var{0}}{\Key{\alpha}{} \circ (\LockSb{\Wk}<\nu_{k-1}>) \dots \circ (\LockSb{\Wk}<\nu_{0}>)}
    }
  }
  \and
  \inferrule{ }{
    \IsNf{\NfTrue}{\Bool}{\True}
    \\
    \IsNf{\NfFalse}{\Bool}{\False}
  }
  \and
  \inferrule{
    \IsNe{e}{\Bool}{M}
  }{
    \IsNf{\NfInj{e}}{\Bool}{M}
  }
  \and
  \inferrule{
    \IsTele{\Theta}{\Gamma}
    \\
    \IsNfTy[\ETele{\Gamma}{\Bool}<\ArrId{m}>]{\tau}{A}
    \\
    \IsNe{u}{\Bool}{M}
    \\
    \IsNf{v_1}{\Sb{A}{\ESb{\ISb}{\True}}}{N_1}
    \\
    \IsNf{v_2}{\Sb{A}{\ESb{\ISb}{\False}}}{N_2}
  }{
    \IsNe{\NeBoolRec{\tau}{u}{v_1}{v_2}}{\Sb{A}{\ESb{\ISb}{M}}}{\BoolRec{A}{M}{N_1}{N_2}}
  }
  \and
  \inferrule{
    \IsTele{\Theta}{\Gamma}
    \\
    \IsTy[\LockCx{\Gamma}]{A}<n>
    \\
    \IsTy[\ECx{\Gamma}{A}]{B}
    \\\\
    \IsNf[\ETele{\Theta}{A}]{u}{B}{M}
  }{
    \IsNf{\NfLam{u}}{\Fn{A}{B}}{\Lam{M}}
  }
  \and
  \inferrule{
    \IsTele{\Theta}{\Gamma}
    \\
    \IsTy[\LockCx{\Gamma}]{A}<n>
    \\
    \IsTy[\ECx{\Gamma}{A}]{B}
    \\\\
    \IsNe{e}{\Fn{A}{B}}{M}
    \\
    \IsNf{u}{A}{N}
  }{
    \IsNe{\NeApp{e}{u}}{\Sb{B}{\ESb{\ISb}{N}}}{\App{M}{N}}
  }
  \and
  \inferrule{
    \IsTele{\Theta}{\Gamma}
    \\
    \IsTy[\LockCx{\Gamma}]{A}<n>
    \\
    \IsTy[\ECx{\Gamma}{A}]{B}
    \\\\
    \IsNf{u}{A}{M}
    \\
    \IsNf{v}{\Sb{B}{\ESb{\ISb}{M}}}{N}
  }{
    \IsNf{\NfPair{u}{v}}{\Sig{A}{B}}{\Pair{M}{N}}
  }
  \and
  \inferrule{
    \IsTele{\Theta}{\Gamma}
    \\
    \IsTy[\LockCx{\Gamma}]{A}<n>
    \\
    \IsTy[\ECx{\Gamma}{A}]{B}
    \\\\
    \IsNe{u}{\Sig{A}{B}}{M}
  }{
    \IsNe{\NeProj[1]{u}}{A}{\Proj[1]{M}}
    \\
    \IsNe{\NeProj[2]{u}}{\Sb{B}{\ESb{\ISb}{\Proj[1]{M}}}}{\Proj[2]{M}}
  }
  \and
  \inferrule{
    \IsTele{\Theta}{\Gamma}
    \\
    \IsTy[\LockCx{\Gamma}]{A}<n>
    \\
    \IsNf[\LockTele{\Theta}]{u}{A}<n>{M}
  }{
    \IsNf{\NfMkBox{u}}{\Modify{A}}{\MkBox{M}}
  }
  \and
  \inferrule{
    \IsTele{\Theta}{\Gamma}
    \\
    \IsTy[\LockCx{\Gamma}]{A}<n>
    \\
    \IsNe{e}{\Modify{A}}{M}
  }{
    \IsNf{\NfInj{e}}{\Modify{A}}{M}
  }
  \and
  \inferrule{
    \Mor[\nu]{o}{n}
    \\
    \Mor[\mu]{n}{m}
    \\
    \IsTele{\Theta}{\Gamma}
    \\
    \IsTy[\LockCx{\LockCx{\Gamma}}<\nu>]{A}<o>
    \\
    \IsNe[\LockTele{\Theta}]{u}{\Modify[\nu]{A}}<n>{M}
    \\
    \IsNfTy[\ETele{\Theta}{\Modify[\nu]{A}}]{\tau}{B}
    \\
    \IsNf[\ETele{\Theta}{A}<\mu \circ \nu>]{u}{\Sb{B}{\ESb{\Wk}{\MkBox[\nu]{\Var{0}}}}}
  }{
    \IsNe{\NeLetMod{\mu}{\nu}{\tau}{e}{u}}{\Sb{B}{\ESb{\ISb}{M}}}{\LetMod{M}{N}}
  }
  \and
  \inferrule{
    \IsTele{\Theta}{\Gamma}
    \\
    \IsNe{e}{\Uni}{A}
  }{
    \IsNf{\NfInj{e}}{\Uni}{A}
  }
  \and
  \inferrule{
    \IsTele{\Theta}{\Gamma}
    \\
    \IsNf[\LockTele{\Theta}]{u}{\Uni}<n>{A}
    \\
    \IsNf[\ETele{\Theta}{\Dec{A}}]{v}{\Uni}{B}
  }{
    \IsNf{\NfFnCode{u}{v}}{\Uni}{\FnCode{A}{B}}
  }
  \and
  \inferrule{
    \IsTele{\Theta}{\Gamma}
    \\
    \IsNf{u}{\Uni}{A}
    \\
    \IsNf[\ETele{\Theta}{\Dec{A}}<\ArrId{m}>]{v}{\Uni}{B}
  }{
    \IsNf{\NfProdCode{u}{v}}{\Uni}{\FnCode{A}{B}}
  }
  \and
  \inferrule{
    \IsTele{\Theta}{\Gamma}
    \\
    \IsNf[\LockTele{\Theta}]{u}{\Uni}{A}
  }{
    \IsNf{\NfModifyCode{u}}{\Uni}{\ModifyCode{A}}
  }
  \and
  \inferrule{
    \IsTele{\Theta}{\Gamma}
  }{
    \IsNf{\NfBoolCode}{\Uni}{\BoolCode}
  }
  \and
  \inferrule{
    \IsTele{\Theta}{\Gamma}
    \\
    \IsNe{e}{\Dec{\BoolCode}}{M}
  }{
    \IsNe{\NfDecIso{e}}{\Bool}{\DecIso{M}}
  }
  \and
  \inferrule{
    \IsTele{\Theta}{\Gamma}
    \\
    \IsNf{u}{\Bool}{M}
  }{
    \IsNf{\NfDecIso*{u}}{\Dec{\BoolCode}}{\DecIso*{M}}
  }
  \and
  \inferrule{
    \IsTele{\Theta}{\Gamma}
    \\
    \IsNe{e}{\Uni}{A}
    \\
    \IsNe{f}{\Dec{e}}{M}
  }{
    \IsNf{\NfInj{f}}{\Dec{e}}{M}
  }
  \and
  \inferrule{
    \IsTele{\Theta}{\Gamma}
    \\
    \IsNf{u}{\Uni}{A}
    \\
    \IsNf[\ETele{\Theta}{\Dec{A}}<\ArrId{m}>]{v}{\Uni}{B}
    \\
    \IsNe{e}{\Dec{\SigCode{A}{B})}}{M}
  }{
    \IsNe{\NfDecIso{e}}{\Sig{\Dec{A}}{\Dec{B}}}{\DecIso{M}}
  }
  \and
  \inferrule{
    \IsTele{\Theta}{\Gamma}
    \\
    \IsNf{u}{\Uni}{A}
    \\
    \IsNf[\ETele{\Theta}{\Dec{A}}<\ArrId{m}>]{v}{\Uni}{B}
    \\
    \IsNf{w}{\Sig{\Dec{A}}{\Dec{B}}}{M}
  }{
    \IsNf{\NfDecIso*{w}}{\Dec{\SigCode{A}{B}}}{\DecIso*{M}}
  }
  \and
  \inferrule{
    \IsTele{\Theta}{\Gamma}
    \\
    \IsNf{u}{\Uni}<n>{A}
    \\
    \IsNf[\ETele{\Theta}{\Dec{A}}]{v}{\Uni}{B}
    \\
    \IsNe{e}{\Dec{\FnCode{A}{B})}}{M}
  }{
    \IsNe{\NfDecIso{e}}{\Fn{\Dec{A}}{\Dec{B}}}{\DecIso{M}}
  }
  \and
  \inferrule{
    \IsTele{\Theta}{\Gamma}
    \\
    \IsNf{u}{\Uni}<n>{A}
    \\
    \IsNf[\ETele{\Theta}{\Dec{A}}]{v}{\Uni}{B}
    \\
    \IsNf{w}{\Fn{\Dec{A}}{\Dec{B}}}{M}
  }{
    \IsNf{\NfDecIso*{w}}{\Dec{\FnCode{A}{B}}}{\DecIso*{M}}
  }
  \and
  \inferrule{
    \IsTele{\Theta}{\Gamma}
    \\
    \IsTm[\LockCx{\Gamma}]{A}{\Uni}<n>
    \\
    \IsNe{e}{\Dec{\ModifyCode{A}}}{M}
  }{
    \IsNe{\NfDecIso{e}}{\Modify{\Dec{A}}}{\DecIso{M}}
  }
  \and
  \inferrule{
    \IsTele{\Theta}{\Gamma}
    \\
    \IsTm[\LockCx{\Gamma}]{A}{\Uni}<n>
    \\
    \IsNf{u}{\Modify{\Dec{A}}}{M}
  }{
    \IsNf{\NfDecIso*{u}}{\Dec{\ModifyCode{A}}}{\DecIso*{M}}
  }
\end{mathparpagebreakable}

\section{A full definition of an internal \MTT{} cosmos}
\label{sec:appendix:cosmos}

We present the full definition of an internal \MTT{} cosmos. The fact that these constants
correspond to the exist in \eg{} the syntactic \MTT{} cosmos follows from unfolding the internal
language this signature is presented in and observing that it is identical to
\cref{def:cosmoi:structured-cosmos} in cosmoi rich enough to use \MTT{} as an internal language.

\subsection{Term and type sorts}
\label{sec:appendix:tm-term}

\begin{align*}
  \Ty{m} &: \Uni\\
  \Tm{m} &: \Ty{m} \to \Uni
\end{align*}

\subsection{Dependent sums}
\label{sec:appendix:tm-sums}

\begin{align*}
  \SigConst &: \brackets{\Sum{A : \Ty{m}} \Tm{m}(A) \to \Ty{m}} \to \Ty{m}
  \\
  \alpha_{\SigConst} &: (A : \Ty{m})(B : \Tm{m}(A) \to \Ty{m}) \to
  \brackets{\Tm{m}(\SigConst(A, B)) \cong \Sum{a : \Tm{m}(A)} \Tm{m}(B(a))}
\end{align*}

\subsection{Dependent products}
\label{sec:appendix:tm-prod}

\begin{align*}
  \PiConst &: \brackets{\Sum{A : \Modify{\Ty{n}}} \LetMod{A}[A]{\brackets{\Modify{\Tm{n}(A)} \to \Ty{m}}}[1]} \to \Ty{m}
  \\
  \alpha_{\PiConst} &: \DeclVar{A}{\Ty{n}} \to (B : \Modify{\Tm{n}(A)} \to \Ty{m}) \to\\
  &\qquad \brackets{\Tm{m}(\PiConst(\MkBox{A}, B)) \cong \Prod{a : \Modify{\Tm{n}(A)}} \Tm{m}(B(a))}
\end{align*}

\subsection{Booleans}
\label{sec:appendix:tm-bool}

\begin{align*}
  \BoolConst &: \Ty{m}
  \\
  \TrueConst,\FalseConst &: \Tm{m}(\BoolConst)
  \\
  \IfConst &: (A : \Tm{m}(\BoolConst) \to \Ty{m}) \to\\
  &\qquad \Tm{m}(A(\TrueConst)) \to \Tm{m}(A(\FalseConst)) \to\\
  &\qquad (b : \Tm{m}(\BoolConst)) \to \Tm{m}(A(b))
  \\
  \_ &: (A : \Tm{m}(\BoolConst) \to \Ty{m})(t : \Tm{m}(A(\TrueConst)))(f : \Tm{m}(A(\FalseConst))) \to\\
  &\qquad (\IfConst(A, t, f, \TrueConst) = t) \times (\IfConst(A, t, f, \FalseConst) = f)
\end{align*}

\subsection{Modal types}
\label{sec:appendix:tm-modality}

\begin{align*}
  \ModConst &: \DeclVar{A}{\Ty{n}} \to \Ty{m}
  \\
  \ModIntroConst &: \DeclVar{A}{\Ty{n}}<\mu>\DeclVar{a}{\Tm{n}(A)} \to \Tm{m}(\ModConst(A))
  \\
  \ModElimConst &:
  \DeclVar{A}{\Ty{n}}<\nu\circ\mu>\brackets{\DelimMin{1}B : \Modify{\Tm{n}(\ModConst(A))} \to \Ty{o}} \to\\
  &\qquad \brackets{\DelimMin{1}\DeclVar{x}{\Tm{n}(A)}<\nu\circ\mu> \to \Tm{o}(B(\ModIntroConst(A, x)))} \to\\
  &\qquad \DeclVar{a}{\Tm{m}(\ModConst(A))}<\nu> \to\\
  &\qquad \Tm{o}(B(a))\\
  \_ &:
  \DeclVar{A}{\Ty{n}}<\nu\circ\mu>\brackets{\DelimMin{1}B : \Modify{\Tm{n}(\ModConst(A))} \to \Ty{o}} \to\\
  &\qquad \brackets{\DelimMin{1}b : \DeclVar{x}{\Tm{n}(A)}<\nu\circ\mu> \to B(\ModIntroConst(A, x))} \to\\
  &\qquad \DeclVar{a}{\Tm{n}(A)}<\nu\circ\mu> \to \ModElimConst(A, B, b, \ModIntroConst(A, a)) = b(a)
\end{align*}

\subsection{Universe \'a la Tarski}
\label{sec:appendix:tm-universe}

\begin{align*}
  \UniConst &: \Ty{m}
  \\
  \DecConst &: \Tm{m}(\UniConst) \to \Ty{m}
  \\
  \SigCodeConst &: \brackets{\Sum{A : \Tm{m}(\UniConst)} \Tm{m}(\DecConst(A)) \to \Tm{m}(\UniConst)} \to \Tm{m}(\UniConst)
  \\
  \PiCodeConst &:
  \brackets{\Sum{A : \Modify{\Tm{n}(\UniConst)}} \LetMod{A}[A]{\brackets{\Modify{\Tm{n}(\DecConst(A))} \to \Ty{m}}}[1]}
  \to \Tm{m}(\UniConst)
  \\
  \BoolCodeConst &: \Tm{m}(\UniConst)
  \\
  \ModCodeConst &: \Modify{\Tm{n}(\UniConst)} \to \Tm{m}(\UniConst)
  \\
  \DecIsoConst_{\SigCodeConst} &:
  (A : \Tm{m}(\UniConst))(B : \Tm{m}(\DecConst(A)) \to \Tm{m}(\UniConst)) \to\\
  &\qquad \Tm{m}(\DecConst(\SigCodeConst(A,B,))) \cong \Tm{m}(\SigConst(\DecConst(A), \DecConst \circ B))
  \\
  \DecIsoConst_{\PiCodeConst} &:
  \DeclVar{A}{\Tm{m}(\UniConst)}(B : \Fn{\Tm{n}(\DecConst(A))}{\Tm{m}(\UniConst)}) \to\\
  &\qquad \Tm{m}(\DecConst(\PiCodeConst(A,B))) \cong \Tm{m}(\PiConst(\DecConst(A), \DecConst \circ B))
  \\
  \DecIsoConst_{\BoolCodeConst} &: \Tm{m}(\DecConst(\BoolCodeConst)) \cong \Tm{m}(\BoolConst)
  \\
  \DecIsoConst_{\ModCodeConst} &:
  \DeclVar{A}{\Tm{m}(\UniConst)} \to
  \Tm{m}(\DecConst(\ModCodeConst(A))) \cong \Tm{m}(\ModConst(\DecConst(A)))
\end{align*}

\subsection{Internal and external \MTT{} cosmoi coincide}

Let us fix a cosmos $F : \VV$ (see \cref{def:cosmoi:cosmos}) such that $F$ supports a model of
\MTT{} where types are interpreted (up to equivalence) by families of objects and a $\Modify{-}$ is
interpreted by $F(\mu)$. For instance, the syntactic \MTT{} cosmos defined by $F(m) = \PSH{\Cx{m}}$.

\begin{theorem}
  \label{thm:appendix:unfold}
  The cosmos $F$ is an \MTT{} cosmos (see \cref{def:cosmoi:structured-cosmos}) precisely when the
  internal language of $F$ supports the constants of an internal \MTT{} cosmos.
\end{theorem}
\begin{proof}
  First, observe that unfolding the constants $\Ty{m}$ and $\Tm{m}$ into the model in $F$ gives an
  object $\TY{m} = \Interp{\Ty{m}}$ and a family over $\TY{m}$ given by $\El{m} =
  \Interp{\Tm{m}}$. We call write $\EL{m}$ for the domain of $\El{m}$. Moreover, this transformation
  is a bijection: every family $\El{m}$ gives rise to a pair of $\Ty{m}$ and $\Tm{m}$.

  Showing that the remain constants induce the structure of an \MTT{} cosmos is a standard exercise
  in unfolding the interpretation of \MTT{} into a presheaf topos. We show only the representative
  case of modalities.

  First, observe that the constant $\ModConst : \DeclVar{A}{\Ty{n}} \to \Ty{m}$ is precisely
  determined by a morphism $\Mor[M]{F(\mu)(\TY{n})}{\TY{m}}$, using the fact that modalities in the
  model of \MTT{} are interpreted by $F(\mu)$. Similarly,
  $\ModIntroConst : \DeclVar{A}{\Ty{n}}<\mu>\DeclVar{a}{\Tm{n}(A)} \to \Tm{m}(\ModConst(A))$ is
  precisely equivalent to a map $\Mor[m]{F(\mu)(\EL{n})}{\EL{m}}$ such that
  $\El{m} \circ m = M \circ F(\mu)(\El{n})$.

  The final equivalence between the elimination constant $\ModElimConst$ and the lifting structure
  from \cref{def:cosmoi:structured-cosmos} is identical similar: unfolding $\ModElimConst$ and its
  equation amounts to an internal lifting structure.
\end{proof}

\section{Neutral and normal forms, internally}
\label{sec:appendix:hoas}

We require that all normal and neutral forms become equal to their counterparts in the syntactic
internal \MTT{} cosmos (\cref{sec:appendix:cosmos,sec:prereq}) under the assumption $z : \Prop$. We
avoid repeatedly stating this in the specifications of normals and neutrals that follows.

\subsection{Normal types}

\begin{align*}
  \CPi &: \DeclVar{A}{\NfTy{n}}(B : \Fn{\Vars{n}(A)}{\NfTy{m}}) \to \NfTy{m}
  \\
  \CSig &: (A : \NfTy{m})(B : \Vars{m}(A) \to \NfTy{m}) \to \NfTy{m}
  \\
  \CBool &: \NfTy{m}
  \\
  \CModify &: \Modify{\NfTy{n}} \to \NfTy{m}
\end{align*}

\subsection{Dependent products}

\begin{align*}
  \CLam &: \DeclVar{A}{\Open[z] \Ty{n}(z)}(B : \Fn{\Open[z] \Tm{n}(z, A(z))}{\Open[z] \Ty{m}{z}})\\
  &{} \to (\FnV{a}{\Vars{n}(A)}{\Nf{m}(B(a))})\\
  &{} \to \Nf{m}(\PiConst(A,B))
  \\
  \CApp &: \DeclVar{A}{\Open[z] \Ty{n}(z)}(B : \Fn{\Open[z] \Tm{n}(z, A(z))}{\Open \Ty{m}})\\
  &{} \to \Ne{m}(\PiConst(A,B))\\
  &{} \to (a : \Nf{m}(A))\\
  &{} \to \Ne{m}(B(\Unit(a)))
\end{align*}

\subsection{Dependent sums}

\begin{align*}
  \CPair &: (A : \Open[z] \Ty{m}(z))(B : \Open[z] \Tm{m}(z, A(z)) \to \Open[z] \Ty{m}(z))\\
  &{} \to (a : \Nf{m}(A)) \to \Nf{m}(B(\Unit(a)))\\
  &{} \to \Nf{m}(\SigConst(A, B))
  \\
  \CProj{0} &: (A : \Open[z] \Ty{m}(z))(B : \Open[z] \Tm{m}(z, A(z)) \to \Open[z] \Ty{m}(z))\\
  &{} \to \Ne{m}(\SigConst(A, B))\\
  &{} \to \Ne{m}(A)
  \\
  \CProj{1} &: (A : \Open[z] \Ty{m}(z))(B : \Open[z] \Tm{m}(z, A(z)) \to \Open[z] \Ty{m}(z))\\
  &{} \to (p : \Ne{m}(\SigConst(A, B)))\\
  &{} \to \Ne{m}(B(\Unit(\CProj{0}(p))))
\end{align*}

\subsection{Booleans}

\begin{align*}
  \CNfInj &: \Ne{m}(\BoolConst) \to \Nf{m}(\BoolConst)
  \\
  \CTrue,\CFalse &: \Nf{m}(\BoolConst)
  \\
  \CIf &: (A : \Open[z] \Tm{m}(z, \BoolConst(z)) \to \Open[z] \Ty{m}(z))\\
  &{} \to \Nf{m}(A(\TrueConst))
  \to \Nf{m}(A(\FalseConst))
  \to (b : \Ne{m}(\BoolConst)) \to \Ne{m}(A(\eta(b)))
\end{align*}

\subsection{Modal types}

\begin{align*}
  \CNfInj &: \DeclVar{A}{\Open[z] \Ty{n}(z)} \to \Ne{m}(\ModConst{A}) \to \Nf{m}(\ModConst{A})
  \\
  \CMkBox &: \DeclVar{A}{\Open[z] \Ty{n}(z)} \to \Modify{\Nf{n}(A)} \to \Nf{\ModConst{A}}
  \\
  \CLetMod &: \DeclVar{A}{\Open[z] \Ty{n}(z)}<\nu\circ\mu>
  (B : \DeclVar{a}{\Open[z] \Tm{m}(z, \ModConst(z, A(z)))}<\nu> \to \Open[z] \Ty{o}(z))\\
  &{} \to (\DeclVar{a}{\Vars{n}(A)}<\mu\circ\nu> \to \Nf{o}(B(\ModIntroConst(A, \eta(a)))))\\
  &{} \to \DeclVar{a}{\Ne{m}(\ModConst(A))}<\nu> \to \Ne{o}(B(\eta(a)))
\end{align*}

\subsection{Universe \'a la Tarski}

\begin{align*}
  \CUni &: \NfTy{m} \qquad \CDec : \Nf{m}(\UniConst) \to \NfTy{m}
  \\
  \CNfInj &: \Ne{m}(\UniConst) \to \Nf{m}(\UniConst)
  \\
  \CPiCode &: \DeclVar{A}{\Nf{n}(\UniConst)}(B : \Fn{\Vars{n}(\CDec(A))}{\Nf{m}(\UniConst)}) \to \Nf{m}(\UniConst)
  \\
  \CSigCode &: (A : \Nf{m}(\UniConst))(B : \Vars{m}(\CDec(A)) \to \Nf{m}(\UniConst)) \to \Nf{m}(\UniConst)
  \\
  \CBoolCode &: \Nf{m}(\UniConst)
  \\
  \CModifyCode &: \Modify{\Nf{n}(\UniConst)} \to \Nf{m}(\UniConst)
  \\
  \CDecIso_{\CPiCode} &:
  \DeclVar{A}{\Nf{n}(\UniConst)}(B : \Fn{\Vars{n}(\CDec(A))}{\Nf{m}(\UniConst)})\\
  &\to \Nf{m}(\CPi(\CDec(A), \CDec(B))) \to \Nf{m}(\DecConst(\CPiCode(A,B)))
  \\
  \CDecIso_{\CSigCode} &:
  (A : \Nf{m}(\UniConst))(B : \Vars{m}(\CDec(A)) \to \Nf{m}(\UniConst))\\
  &\to
  \Nf{m}(\CSig(\CDec(A), \CDec(B))) \to \Nf{m}(\DecConst(\CSigCode(A,B)))
  \\
  \CDecIso*_{\CPiCode} &:
  \DeclVar{A}{\Nf{n}(\UniConst)}(B : \Fn{\Vars{n}(\CDec(A))}{\Nf{m}(\UniConst)})\\
  &\to \Ne{m}(\DecConst(\CPiCode(A,B))) \to \Ne{m}(\CPi({\CDec(A)}{}, \CDec(B)))
  \\
  \CDecIso*_{\CSigCode} &:
  (A : \Nf{m}(\UniConst))(B : \Vars{m}(\CDec(A)) \to \Nf{m}(\UniConst))\\
  &\to \Ne{m}(\DecConst(\CSigCode(A,B))) \to \Ne{m}(\CSig(\CDec(A), \CDec(B)))
  \\
  \CDecIso_{\CBoolCode} &: \Nf{m}(\BoolConst) \to \Nf{m}(\DecConst(\BoolCodeConst))
  \\
  \CDecIso_{\CModifyCode} &: \DeclVar{A}{\Nf{n}(\UniConst)} \to \Nf{m}(\ModConst(A)) \to \Nf{m}(\DecConst(\ModCodeConst(A)))
  \\
  \CDecIso*_{\CBoolCode} &: \Ne{m}(\DecConst(\BoolCodeConst)) \to \Ne{m}(\BoolConst)
  \\
  \CDecIso*_{\CModifyCode} &: \DeclVar{A}{\Nf{n}(\UniConst)} \to \Ne{m}(\DecConst(\ModCodeConst(A))) \to \Ne{m}(\ModConst(A))
\end{align*}

\subsection{Interpreting neutral and normal forms in the glued model}

Substantiating these constants in the glued model (see \cref{thm:prereq:model}) relies on unfolding
the constants described above into this model and showing that the expected normal or neutral form
from \cref{sec:renamings} can be used. These proofs follow \textcite{hofmann:1999}, and again the
main subtlety is showing that binders in constants such as $\CPi$ are correctly interpreted.

For each normal and neutral form, the requirement that constant lie over the appropriate constant
from the syntactic cosmos forces the ``syntactic portion'' of each constant.  More precisely,
consider again $\CPi$ and set $I = \DeclVar{A}{\NfTy{n}} \times \Vars{n}(A)$ and
$I' = \DeclVar{A}{\Ty{n}} \times \Tm{n}(A)$. The definition of $\CPi$ is a commuting
square of the following shape in $\PSH{\Ren{m}}$:
\[
  \DiagramSquare{
    width = 5cm,
    ne = \Interp{\NfTy{m}}^1,
    nw = \Interp{I}^1,
    se = \InvRen{\Interp{\Ty{m}}},
    sw = \InvRen{\Interp{I'}},
    west = \Interp{I},
    east = \Interp{\NfTy{m}},
    north = \Interp{\CPi}^1,
    south = \Interp{\CPi}^0,
  }
\]
The bottom map of this square $\Interp{\CPi}^0$, moreover, must be $\InvRen{\PiConst}$ in order to
ensure that $\CPi = \PiConst(z)$ with $z : U$. It remains to define the top arrow of this
diagram. This is complicated by the difference between the interpretation of dependent products in
$\InterpGl{m}$ and $\PSH{\Ren{m}}$. In particular, while the ``downstairs'' portion of a dependent
product in $\InterpGl{m}$ is a dependent product in $\PSH{\Cx{m}}$, the upstairs portion is not a
dependent product in $\PSH{\Ren{m}}$.

In order to define $\Interp{\CPi}^1$, it suffices to fix a morphism
$\Mor[x]{\Yo{\Theta}}{\Interp{\NfTy{m}}^1}$ and define the action of $\Interp{\CPi}^1$ on it. First,
observe that we can extend $x$ to the following commuting square:
\[
  \DiagramSquare{
    width = 5cm,
    nw = \Yo{\Theta},
    sw = \InvRen{\Yo{\EmbRen(\Theta)}},
    ne = \Interp{I}^0,
    se = \InvRen{\Interp{I'}},
    east = \Interp{I},
    north = x,
    south = \InvRen{\tilde{x}},
  }
\]
By universal property, therefore, we can decompose $x$ into a pair:
\begin{gather*}
  \DiagramSquare{
    width = 5cm,
    nw = \Yo{\LockRen{\Theta}},
    sw = \InvRen[n]{\Yo{\LockCx{(\EmbRen(\Theta))}}},
    ne = \Interp{\NfTy{n}}^1,
    se = \InvRen[n]{\Interp{\Ty{n}}},
    east = \Interp{\NfTy{n}},
    north = x_0,
    south = \InvRen[n]{\tilde{x}_0},
  }
  \\
  \DiagramSquare{
    width = 5cm,
    nw = \Yo{\ERen{\Theta}{x_0}},
    sw = \InvRen{\Yo{\EmbRen(\ECx{\Theta}{\tilde{x}_0})}},
    ne = \Interp{\NfTy{m}}^1,
    se = \InvRen{\Interp{\Ty{m}}},
    east = \Interp{\NfTy{n}},
    north = x_1,
    south = \InvRen{\tilde{x}_1},
  }
\end{gather*}

We may now define $\CPi(x) = \NfFn{x_0}{x_1}$. The naturality of this assignment, as well as the
fact that it commutes appropriately, are direct computations.


\printbibliography

\end{document}